\documentclass[runningheads]{llncs}
\usepackage{graphicx,comment}
\usepackage{amsmath}
\usepackage{tikz}
\usepackage[T1]{fontenc}

\newcommand \dia{\hfill{$\diamond$}}

\newcommand{\NP}{{\sf NP}}

\begin{document}
\title{Acyclic, Star, and Injective Colouring:\\ Bounding the Diameter\thanks{An extended abstract of this paper will appear in the proceedings of WG 2021.}}
 
\titlerunning{Acyclic, Star and Injective Colouring: Bounding the Diameter}

\author{Christoph Brause\inst{1} \and
Petr Golovach\inst{2}\orcidID{0000-0002-2619-2990} \and
Barnaby Martin\inst{3} \and
Pascal Ochem\inst{4} \and
Dani\"el Paulusma\inst{3}\orcidID{0000-0001-5945-9287}
\thanks{Author supported by the Leverhulme Trust (RPG-2016-258).} 
\and
Siani Smith\inst{3} 
}
\authorrunning{C. Brause et al.}
\institute{TU Bergakademie Freiberg, Germany,
\email{brause@math.tu-freiberg.de}\\ \and
University of Bergen, Norway,
\email{petr.golovach@ii.uib.no}
\and
Durham University, United Kingdom\\
\email{\{barnaby.d.martin,daniel.paulusma,siani.smith\}@durham.ac.uk}\\ \and
CNRS, LIRMM, Universit\'e de Montpellier, France, \email{ochem@lirmm.fr}
}

\maketitle              
\begin{abstract}
We examine the effect of bounding the diameter for a number of natural and well-studied variants of the {\sc Colouring} problem. 
A colouring is acyclic, star, or injective if any two colour classes induce a forest, star forest or disjoint union of vertices and edges, respectively. 
The corresponding decision problems are {\sc Acyclic Colouring}, {\sc Star Colouring} and {\sc Injective Colouring}. The last problem is also known as {\sc $L(1,1)$-Labelling} and we also consider the framework of {\sc $L(a,b)$-Labelling}.
We prove a number of (almost-)complete complexity classifications. In particular, we show that for graphs of diameter at most~$d$, {\sc Acyclic $3$-Colouring} is polynomial-time solvable if $d\leq 2$ but \NP-complete if $d\geq 4$, and {\sc Star $3$-Colouring} is polynomial-time solvable if $d\leq 3$ but \NP-complete for $d\geq 8$. As far as we are aware, {\sc Star $3$-Colouring} is the {\it first} problem that exhibits a complexity jump for some $d\geq 3$. Our third main result is that {\sc $L(1,2)$-Labelling} is \NP-complete for graphs of diameter~$2$; we relate the latter problem to a special case of {\sc Hamiltonian Path}.
\end{abstract}

\setcounter{footnote}{0}

\section{Introduction}\label{s-intro}

A natural way of increasing our understanding of \NP-complete problems is to put some restrictions on the input.
For graph problems, this means that we may consider graphs from a class characterized by a special property or parameter.
In particular, {\it hereditary} graph classes have been studied. These are the graph classes closed under vertex deletion. The framework of hereditary graph classes covers many well-known graph classes, including $H$-free graphs (graphs with no induced subgraph isomorphic to some fixed graph $H$), bipartite graphs, chordal graphs, planar graphs, and so on. However, not all natural graph classes studied in the literature are hereditary. Moreover, studying non-hereditary graph classes may also yield new insights in the computational complexity of \NP-complete graph problems. The latter is the goal in this paper, and for this purpose we consider classes of graphs whose diameter is bounded by some constant $d\geq 1$.

\subsection{Bounding the diameter}

The {\it diameter} of a graph~$G$ is the maximum distance between any two vertices of $G$.
For a positive integer $d$, the class of graphs of diameter at most~$d$ is hereditary if and only if $d\leq 1$; in order to see this, note that graphs of diameter~1 are the complete graphs, whereas the path $P_3$ on three vertices has diameter~$2$ but becomes disconnected after removing the middle vertex. 

Many graph problems stay \NP-complete if we bound the diameter, even if we set $d=2$ (note that the case $d=1$ is of limited interest in most problem settings). The reason for this hardness is usually the following: from a general problem instance we can obtain an equivalent instance of diameter~$2$ by adding a {\it dominating}
vertex, that is, a vertex that is made adjacent to all the other vertices of the graph.
For example, this reduction can be used for classical graph problems, such as those of deciding if for a given integer $k$, a graph has a clique of size at most $k$ ({\sc Clique}) or an independent set of size at most $k$ ({\sc Independent Set}). For the {\sc Independent Set} problem, bounding the diameter does not yield any new tractable classes even if the instance is also $H$-free for some graph $H$~\cite{BGMPS} (in this case adding a dominating vertex may violate the $H$-freeness condition).

The simple trick of adding a dominating vertex can also be used for graph partitioning problems. To give a well-known example,
a vertex mapping $c:V\to \{1,2,\ldots,k\}$ is a {\it colouring}, or more specifically, a {\it $k$-colouring} of a graph $G=(V,E)$ if for every edge $uv\in E$ it holds that $c(u)\neq c(v)$. The {\sc Colouring} problem is to decide for a given graph $G$ and integer~$k$, if $G$ has a $k$-colouring, or equivalently, if $V(G)$ can be partitioned into $k$ independent sets. By using the trick, it is readily seen that {\sc Colouring} stays \NP-complete even for graphs of diameter~$2$. 

However, the situation becomes less clear for graph partitioning problems if the upper bound $k$ on the number of partitioning classes is {\it fixed}, that is, no longer part of the input. In this setting, adding a dominating vertex may increase the number of partition classes by one, as is the case for the {\sc Colouring} problem.
If $k$ is fixed, we write $k$-{\sc Colouring} instead. By using another (straightforward) gadget, it follows nevertheless that for $d\geq 2$ and $k\geq 3$, the $k$-{\sc Colouring} problems for graphs of diameter at most~$d$ stays \NP-complete for every pair $(d,k)\notin \{(2,3),(3,3)\}$. In addition, Mertzios and Spirakis~\cite{MS16} gave a highly non-trivial \NP-hardness proof for the case $(3,3)$.
The case $(2,3)$, that is, determining the computational complexity of {\sc $3$-Colouring} for graphs of diameter at most~$2$, is a notorious open problem~\cite{BKM12,BFGP13,DPR21,MPS21,MPS19,MS16,Pa15} (which is not the focus of our paper).

The problem {\sc Near Bipartiteness} is to decide if a graph has a $3$-colouring such that (only) two colour classes induce a forest.
In contrast to the aforementioned problems, this problem is an example of a graph partitioning problem with fixed $k$, for which bounding the diameter to $d=2$ gives us a positive result. Namely, the {\sc Near-Bipartiteness} problem, on graphs of diameter at most~$d$, is polynomial-time solvable if $d\leq 2$~\cite{YY06} and \NP-complete if $d\geq 3$~\cite{BDFJP18}. 

\subsection{Our Focus} 
We consider a number of well-studied and closely related variants of graph colouring (in particular for fixed $k$) and ask:

\medskip
\noindent 
{\it How much does bounding the diameter help for obtaining polynomial-time algorithms for well-known graph colouring variants?}

\medskip
\noindent
 In order to define the variants, we first need to introduce some new terminology.
For $i\in \{1,\ldots,k\}$, the $i$th {\it colour class} of a graph $G=(V,E)$ with a $k$-colouring $c$ is the set $$V_i=\{u\in V\; |\; c(u)=i\}.$$ For $i\neq j$, let $G_{i,j}$ be the (bipartite) subgraph of $G$ induced by $V_i\cup V_j$. If every $G_{i,j}$ is a forest, then $c$ is an {\it acyclic ($k$-)colouring}. For an integer $n\geq 1$, let $P_n$ denote the $n$-vertex path. If every $G_{i,j}$ is a $P_4$-free forest, that is, a disjoint union of stars, then $c$ is a {\it star ($k$-)colouring}. If every $G_{i,j}$ is $P_3$-free, that is, a disjoint union of vertices and edges, then $c$ is an {\it injective ($k$-)colouring}. 
Note that an injective colouring is a star colouring and a star colouring is an acyclic colouring, but the reverse implications might not be true.

The three decision problems, which are to decide for a given graph $G$ and integer $k\geq 1$, if $G$ has an acyclic $k$-colouring, star $k$-colouring or injective $k$-colouring, respectively,
are called {\sc Acyclic Colouring}, {\sc Star Colouring} and {\sc Injective Colouring}, respectively.\footnote{In some papers (for example,~\cite{HKSS02,HRS08,JXZ13}), injective colourings are not necessarily proper, that is, two adjacent vertices may be coloured alike. However, we do {\it not} allow this: as can be observed from the definitions, all colourings considered in our paper are proper.} If $k$ is fixed, then we write {\sc Acyclic $k$-Colouring}, {\sc Star $k$-Colouring} and {\sc Injective $k$-Colouring}.

In another well-studied framework, injective colourings are known as {\it distance-$2$ colourings} and as {\it $L(1,1)$-labellings}. Namely, a colouring of a graph $G$ is injective if the neighbours of every vertex of $G$ are coloured differently, that is, also vertices of distance~$2$ from each other must be coloured differently. 
More generally, a vertex mapping $c:V\to \{1,2,\ldots,k\}$ is an $L(a_1,\ldots,a_p)$-($k$-)labelling if for every two vertices $u$ and $v$ and  every integer $1\leq i\leq p$: if $G$ contains a path of length~$i$ between $u$ and $v$, then $|c(u)-c(v)|\geq a_i$; see also~\cite{Ca11}. If $a_1\geq a_2 \geq \ldots \geq a_p$, the condition is equivalent to  ``if $u$ and $v$ are of distance $i$''. For integers $a_1,\ldots, a_p$ ($p\geq 1$), the {\it distance constrained labelling} problem  {\sc $L(a_1,\ldots,a_p)$-Labelling} is to decide for a given graph $G$ and integer $k$, if $G$ has an $L(a_1,\ldots,a_p)$-$k$-labelling.

All the above problems are \NP-complete, even for very restricted (hereditary) graph classes, see, for example,~\cite{ACKKR04,AZ02,BKTL04,CC86,CMS11,Ko78,LSS18,SMMS14,LR92,Ly11,Ma02,MNRW13,Oc05,SH97,ZKN00} and more recent papers, such as~\cite{BJMPS20,BJMPS21,Ka18,SA20}.\footnote{Some of the old and recent papers in this list also contain tractability results for hereditary graph classes. These classes are not the focus of our paper. However, these papers do illustrate that the colouring variants we study in the paper have a long history.} We refer to the survey paper of Calamoneri~\cite{Ca11} for a large variety of complexity results on distance constrained labelling problems.

Recall from the aforementioned example of {\sc Near-Bipartiteness} that bounding the diameter may yield a change in computational complexity from $d=2$ to $d=3$. 
To illustrate a possible complexity change even better, we consider {\sc $L(a_1,\ldots,a_p)$-$k$-Labelling}. The degree of every vertex of a graph $G$ with an $L(a_1,\ldots,a_p)$-$k$-labelling is at most $k$. Hence, $|V(G)|\leq 1+k+\ldots +k^d$, where $d$ is the diameter of $G$, and we can make the following observation:

\begin{proposition}\label{o-abdk}
For $a_1,\ldots,a_p,d,k\geq 1$, the {\sc $L(a_1,\ldots,a_p)$-$k$-Labelling} problem is constant-time solvable for graphs of diameter at most~$d$.
\end{proposition}

\noindent
Note that Proposition~\ref{o-abdk} implies that for every $k\geq 1$ and $d\geq 1$, {\sc Injective $k$-Colouring} (the case where $p=2$ and $a_1=a_2=1$) is constant-time solvable for graphs of diameter at most~$d$. However, if $k$ is part of the input, then {\sc Injective Colouring} is \NP-complete even for graphs of diameter at most~$2$, and the same holds for 
{\sc Acyclic Colouring} and {\sc Star Colouring}. This follows immediately from the ``dominating vertex'' trick. 

\subsection{Our Results} 

Motivated by Proposition~\ref{o-abdk} we first consider the problems {\sc Acyclic $k$-Colouring} and {\sc Star $k$-Colouring} for graphs of bounded diameter.
In Sections~\ref{s-acyclic} and \ref{s-star}, respectively, we prove the following two almost-complete dichotomies; note that the case where $k\leq 2$ is trivial.

\begin{theorem}\label{t-main1}
For $d\geq 1$ and $k\geq 3$,  
{\sc Acyclic $k$-Colouring} on graphs of diameter at most~$d$ is
polynomial-time solvable if $d=1$, $k\geq 4$ or $d\leq 2$, $k=3$ and \NP-complete if $d\geq 2$, $k\geq 4$ or  $d\geq 4$, $k=3$.
\end{theorem}

\begin{theorem}\label{t-main2}
For $d\geq 1$ and $k\geq 3$,  
{\sc Star $k$-Colouring} on graphs of diameter at most~$d$ is
polynomial-time solvable if $d=1$, $k\geq 4$ or $d\leq 3$, $k=3$ and \NP-complete if $d\geq 2$, $k\geq 4$ or $d\geq 8$, $k=3$.
\end{theorem}

\noindent
Theorem~\ref{t-main1} leaves only open the case where $d=k=3$, that is, {\sc Acyclic $3$-Colouring} for graphs of diameter at most~$3$.
Theorem~\ref{t-main2} leaves only open four cases where $4\leq d\leq 7$ and $k=3$, that is, {\sc Star $3$-Colouring} for graph of diameter at most $d$, where $d\in \{4,5,6,7\}$. 
The case $d=3$, $k=4$ in Theorem~\ref{t-main1} follows from a stronger result that we prove. Namely, we show (in Section~\ref{s-acyclic}) that {\sc Acyclic $3$-Colouring} is \NP-complete even for triangle-free 2-degenerate graphs of diameter at most~$4$ (a graph is {\it $2$-degenerate} if every subgraph of it has a vertex of degree at most~$2$). This is a reduction using a new gadget. The other hardness results in Theorems~\ref{t-main1} and~\ref{t-main2} are obtained by straightforward reductions.
Our new polynomial-time result in Theorem~\ref{t-main1} is for the case where $d=2$ and $k=3$. We obtain this result by a careful analysis of the structure of the diameter-$2$ yes-instances of {\sc Acyclic $3$-Colouring}. 

The {\it main result} of our paper is the new polynomial-time result in Theorem~\ref{t-main2}, which is for the case where $d=3$ and $k=3$.
Namely, we prove that {\sc Star $3$-Colouring} can be solved in polynomial time for graphs of diameter at most~$3$. As we also show that
{\sc Star $3$-Colouring} is \NP-complete for graphs of diameter at most~$8$, we have a complexity jump between $d=3$ and $d=8$. We are not aware of any other graph problems exhibiting a jump after $d=3$.

In order to prove our main result we deduce some structural and easy-to-verify properties of diameter~$3$ yes-instances of {\sc Star $3$-Colouring}. This analysis allows us to preprocess the input graph in order to make its structure simpler. Consequently, we can reduce an instance of it to a polynomial number of instances of $2$-{\sc List Colouring}. This is a standard step in graph colouring, as {\sc $2$-List Colouring} is known to  be polynomial-time solvable. Nevertheless some problem-specific technical analysis is needed in order to perform this step. Moreover, in contrast to classical graph colouring, we are not done yet as we need the star colouring property to hold as well. However, we show that  this property can indeed be preserved by a small blow-up of the created instances of $2$-{\sc List Colouring}.

\medskip
\noindent
Finally, we consider {\sc $L(a,b)$-Labelling} for the most studied values of $(a,b)$, namely when $1\leq a\leq b\leq 2$. 
Due to Proposition~\ref{o-abdk}, we now assume that $k$ is part of the input.
Every two non-adjacent vertices in a graph~$G$ of diameter~$2$ have a common neighbour. Hence, an $L(1,1)$-labelling of $G$ colours each vertex uniquely. Therefore, {\sc $L(1,1)$-Labelling} is trivial for graphs of diameter at most $2$. The {\sc L(1,1)-Labelling} problem is still \NP-complete for graphs of diameter at most $3$, as it is \NP-complete for split graphs~\cite{BKTL04}
(a graph is {\it split} if its vertex set can be partitioned into a clique and independent set, and connected split graphs have diameter at most~$3$). 
Griggs and Yeh~\cite{GY92} proved that {\sc $L(2,1)$-Labelling} is \NP-complete for graphs of diameter at most~$2$ by pinpointing a relation with {\sc Hamiltonian Path}.  

The above leaves us with exactly one case, namely {\sc $L(1,2)$-Labelling} for graphs of diameter at most~$2$. In Section~\ref{s-injective}, we prove that this case is \NP-complete as, by making a connection to {\sc Hamiltonian Path} as well. That is, we observe that 
 an $n$-vertex graph $G$ of diameter~$2$ has an $L(1,2)$-$n$-labelling if and only if $G$ has a Hamiltonian path, no edge of which is contained in a triangle. This observation allows us to adapt the construction of Krishnamoorthy~\cite{Kr75} for proving that {\sc Hamilton Cycle} is \NP-complete for bipartite graphs.

To summarize, we obtained the following dichotomy for {\sc $L(a,b)$-Labelling} $(a,b)\in \{1,2\})$ restricted to graphs of bounded diameter: 

\begin{theorem}\label{t-main3}
For $a,b\in \{1,2\}$ and $d\geq 1$, {\sc $L(a,b)$-Labelling} on graphs of diameter at most~$d$ is 
\begin{itemize}
\item polynomial-time solvable if $a=b$ and $d\leq 2$, or $d=1$; and
\item \NP-complete if either $a=b$ and $d\geq 3$, or $a\neq b$ and $d\geq 2$.
\end{itemize}
\end{theorem}

\section{The Proof of Theorem~\ref{t-main1}}\label{s-acyclic}

We first prove the following result. 
In the proof of this result we let the graph $2P_2$ denote the disjoint union of two $2$-vertex paths, that is, $2P_2$ is the graph with vertices $x_1,x_2,y_1,y_2$ and edges $x_1x_2$ and $y_1y_2$.

\begin{lemma}\label{l-ad2}
{\sc Acyclic $3$-Colouring} is polynomial-time solvable for graphs of diameter at most~$2$.
\end{lemma}

\begin{proof}
Let $G$ be a graph of diameter at most~$2$ with $n$ vertices and $m$ edges.
If $n\leq 24$ or $G$ has diameter~1, we check if $G$ has an acyclic $3$-colouring in linear time. We can check in $O(mn)$ time if there is a vertex $u$ such that $G-u$ is a forest. If so, then $G$ has an acyclic $3$-colouring (give $u$ colour~$1$ and use colours~$2$ and~$3$ for $G-u$). 
Now assume that $G$ has at least 25 vertices and diameter 2 and $G-u$ is not a forest for every $u\in V$.
We show a crucial claim:

\medskip
\noindent
{\it Claim 1. If $G=(V,E)$ is a yes-instance of {\sc Acyclic $3$-Colouring}, then there exists a set $S\subseteq V$ with $|S|\leq 1$ such that $G-S$ contains an induced $2P_2$, say with edges $u_1v_1$ and $u_2v_2$, for which the following two conditions hold:
\begin{itemize}
\item $(N_G(\{u_1,v_1\})\cap N_G(\{u_2,v_2\}))\setminus S$ is a colour class of an acyclic $3$-colouring~$c$ of $G$,~and
 \item the other two colour classes of $c$ induce a subgraph of $G$ with at most two 
connected 
 components.
\end{itemize}}

\noindent
{\it Proof of Claim~1.}
Assume $G$ has an acyclic $3$-colouring $c$ with colour classes $X_1$, $X_2$ and $X_3$ with $1\leq |X_1|\leq |X_2|\leq |X_3|$. 
Let $F=G-X_3$.
Note that $F$ is a forest, as $c$ is acyclic, and that $|X_3|\geq 9$, as $|V|\geq 25$. 

Let $U$ be the set of isolated vertices of $F$, and let $M$ be a matching of maximum size in $F$ such that each edge of $M$ is incident to a leaf of $F$. As $G$ has diameter~2, every vertex of $X_3$ is adjacent to every vertex of $U$. By the same reason, every vertex in $X_3$ must be adjacent to a leaf of $F$ or else to its parent in $F$. Hence, every vertex of $X_3$ is also adjacent to at least one end-vertex of every edge in $M$. 

As $c$ is acyclic, there do not exist sets $T\subseteq X_1\cup X_2$ with $|T|=3$ and $X_3'\subseteq X_3$ with $|X_3'|=2$
such that every vertex of $T$ is adjacent to every vertex of $X_3'$. Combining this observation with the above two yields that $|U|\leq 2$. In addition to this we use the pigeonhole principle and the fact that $|X_3|\geq 9$ to find that $|M|\leq 2$ and $1\leq |M|+|U|\leq 2$ (recall that $F$ is nonempty and thus $1\leq |M|+|U|$).

If $|V(F)\setminus U|=2$ and $U=\emptyset$, or $|M|=0$ and $1\leq |U|\leq 2$, then one of $X_1$ or $X_2$ has size~1, say $X_1=\{u\}$. Then $V\setminus \{u\}=X_2\cup X_3$ is a forest, as $c$ is acyclic, a contradiction. Hence, as $1\leq |M|+|U|\leq 2$, we find that  $|M|=2$ and $|U|=0$ or else that $F-U$ is a star with at least two leaves (so $|M|=1$) and $|U|=1$. However, the latter case is not possible. To see this, let $U=\{u\}$, with say $u\in X_1$, and let $V(F)\setminus U$ be a star with center~$a$ and leaves $b_1,\ldots,b_r$ for some $r\geq 2$.  If $a\in X_2$, then $X_2=\{a\}$ and $V(F)\setminus \{a\}=X_1\cup X_3$ induces a forest, a contradiction.
Hence, $X_1=\{a,u\}$ and $X_2=\{b_1,\ldots,b_r\}$. As every vertex of $X_3$ is adjacent to~$u$ and $c$ is acyclic, at most one vertex of $X_3$ is adjacent to $a$. Consequently, every other vertex of $X_3$ is adjacent to every $b_i$, as $G$ has diameter~2. However, as $c$ is acyclic and $r\geq 2$, at most one vertex of $X_3$ can be adjacent to every $b_i$. Then $X_3$ has at most two vertices, contradicting the fact that $|X_3|\geq 9$. We conclude that $|M|=2$ and $|U|=\emptyset$.

Let $M=\{u_1v_1,u_2v_2\}$. As $F$ is a forest, $N_G(\{u_1,v_1\})\cap N_G(\{u_2,v_2\})$ contains at most one vertex~$s$. Let $S=\{s\}$ if such a vertex $s$ exists and let $S=\emptyset$ otherwise. Since every vertex of $X_3$ is adjacent to an end-vertex of every edge of $M$, we find that $X_3=(N_G(\{u_1,v_1\})\cap N_G(\{u_2,v_2\}))\setminus S$. Moreover, as $|M|\leq 2$ and $U=\emptyset$, we have that $F=G-X_3$ has at most two connected components. Hence, we have proven Claim~1. \dia

\medskip
\noindent
We consider all possible $O(m^2n)$ selections of an induced $2P_2$ with edges $u_1v_1$ and $u_2v_2$ and a set of vertices $S\subseteq V(G)\setminus\{u_1,v_1,u_2,v_2\}$ of size at most~$1$. For each choice, we find, in $O(n+m)$ time, the set  $X_3=(N_G(\{u_1,v_1\})\cap N_G(\{u_2,v_2\}))\setminus S$. Then we verify in $O(n+m)$ time if $X_3$ is an independent set and $F=G-X_3$ is a forest with at most two connected components. If this is not the case, we discard the current choice. Otherwise we continue. As $F$ has at most two 
connected
 components, $F$ has at most two $2$-colourings (up to symmetry). For each $2$-colouring of $F$, we check in $O(n+m)$ time if its two colour classes $X_1$ and $X_2$ together with $X_3$ yield an acyclic $3$-colouring of $G$. 

The correctness of the algorithm immediately follows from Claim~1. The running time is $O(m^3n)$. The latter can be improved to $O(n^4)$ 
as follows. We first check whether $m\leq 2n-3$ holds. We are allowed to do so, as the latter is a necessary condition: if $c$ is an acyclic $3$-colouring of $G$ with colour classes $X_1$, $X_2$ and $X_3$, then 
$G$ is the union of the forests $G[X_1\cup X_2]$, $G[X_1\cup X_3]$ and $G[X_2\cup X_3]$, and thus $m\leq |X_1\cup X_2|-1+|X_1\cup X_3|-1+|X_2\cup X_3|-1=2n-3$. \qed
\end{proof}

\noindent
We complement the previous, algorithmic result by a new hardness result.

\begin{theorem}\label{lem:hard-acycl}
{\sc Acyclic $3$-Colouring} is \NP-complete on triangle-free $2$-degenerate graphs of diameter at most~$4$.
\end{theorem}

\begin{proof}
We reduce from the problem {\sc Near Bipartiteness}, which is known to be \NP-complete~\cite{BLS98}.
This problem asks whether a graph admits an $(I,F)$-partition, that is, a vertex partition into a forest and an independent set.
Let $G$ be an instance of {\sc Near Bipartiteness}. Without loss of generality, the minimum degree of $G$ is at least~$3$.
We construct the graph~$G'$ from $G$ as follows.
We subdivide every edge of $G$ to obtain a bipartite graph containing the \emph{old} vertices of degree at least $3$ in one part and the \emph{new} vertices of degree $2$ in the other part.
We add a vertex $x$.
For every old vertex $v_o$, we add two vertices of degree~$2$ adjacent to $v_o$ and $x$.
For every new vertex $v_n$, we add three vertices of degree~$2$ adjacent to $v_n$ and $x$.
Figure~\ref{edge} shows the graph $G'$ if $G$ is an edge $uv$.

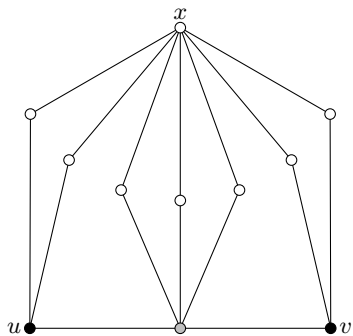
\begin{figure}[htbp]
\begin{center}
\begin{tikzpicture}
%% Edges
\draw (-2,-4)--(2,-4);
\foreach \x in {1,2}{
	\draw (0,0)--({(270-80+160/8*\x)}:2.3) -- (-2,-4);
}
\foreach \x in {3,4,5}{
	\draw (0,0)--({(270-80+160/8*\x)}:2.3) -- (0,-4);
}
\foreach \x in {6,7}{
	\draw (0,0)--({(270-80+160/8*\x)}:2.3) -- (2,-4);
}
%% middle vertices
\foreach \x in {1,2}{
	\draw[fill=white] ({(270-80+160/8*\x)}:2.3) circle[radius=2pt] ;
}
\foreach \x in {3,4,5}{
	\draw[fill=white] ({(270-80+160/8*\x)}:2.3) circle[radius=2pt] ;
}
\foreach \x in {6,7}{
	\draw[fill=white] ({(270-80+160/8*\x)}:2.3) circle[radius=2pt] ;
}
%% u
\draw[fill=black] (-2,-4) circle[radius=2pt] ;
\node[left] at (-2,-4){$u$};
%% vertex between u and v
\draw[fill=lightgray] (0,-4) circle[radius=2pt] ;
%% v
\draw[fill=black] (2,-4) circle[radius=2pt] ;
\node[right] at (2,-4){$v$};
%% x
\draw[fill=white] (0,0) circle[radius=2pt] ;
\node[above] at (0,0){$x$};
\end{tikzpicture}
\caption{The graph $G'$ if $G$ consists of the edge $uv$. The vertices $u$ and $v$ are the old vertices of $G'$, and their common neighbour is the new vertex of $G'$.}
\label{edge}
\end{center}
\end{figure}

By construction, $G'$ is triangle-free, $2$-degenerate, and its diameter is at most~$4$ since every vertex is at distance at most $2$ from $x$.
We claim that $G$ has an $(I,F)$-partition if and only if $G'$ has an acyclic $3$-colouring.

\medskip
\noindent
First suppose that $G$ has an $(I,F)$-partition.
We assign colour $0$ to $x$.
We assign colour $2$ to every new vertex.
We assign colour $1$ to every vertex adjacent to $x$ and a new vertex.
For every vertex in $F$, we assign colour $1$ to the corresponding old vertex $v_F$ and we assign colour $2$ to the two vertices adjacent to $x$ and $v_F$.
For every vertex in $I$, we assign colour $0$ to the corresponding old vertex $v_I$ and we assign colours $1$ and $2$ to the two vertices adjacent to $x$ and $v_I$.

We claim that this $3$-colouring of $G'$ is acyclic. In order to see this, let $G'_{i,j}$ be the induced subgraph of $G'$ with vertices coloured $i$ and $j$ for $i\neq j$.
We first consider $G'_{0,1}$. The connected component of $G'_{0,1}$ that contains $x$ is a tree in which every vertex except $x$
has degree $1$ or $2$ and is at distance at most~$2$ from $x$. Every other connected component of $G'_{0,1}$ consists of an isolated (old) vertex coloured $1$.

Now consider $G'_{0,2}$. As $I$ is an independent set of $G$, every new vertex in $G'$ is adjacent to at most one vertex with colour~$0$.
Hence, the connected component of $G'_{0,2}$ that contains $x$ is a tree in which every vertex except $x$ has degree $1$ or $2$ and is at distance at most~$3$ from $x$. Every other connected component of $G'_{0,2}$ consists of an isolated (new) vertex coloured~$2$.

Finally consider the graph $G'_{1,2}$. For contradiction, suppose that $G'_{1,2}$ contains a cycle.
This cycle cannot contain a vertex adjacent to $x$ in $G'$, since such a vertex has degree~$1$ in $G'_{1,2}$.
So this cycle alternates between old vertices coloured $1$ and new vertices coloured~$2$.
These old and new vertices correspond respectively to the vertices and the edges of a cycle in $F$, a contradiction.
So $G'$ has an acyclic $3$-colouring.

\medskip
\noindent
Now suppose that $G'$ has an acyclic $3$-colouring~$c$. Say $x$ is coloured $0$. 
For every old vertex coloured $0$, the corresponding vertex in $G$ is assigned to $I$.
Every other vertex of $G$ is assigned to $F$.

For contradiction, suppose that $I$ contains an edge $uv$. Then, $G'$ contains a new vertex $z$ that is adjacent to $u$ and $v$. As $c$ is acyclic, the two common neighbours of $x$ and $u$ are coloured $1$ and $2$, respectively, and the same holds for the two common neighbours of $x$ and $v$.
Then $G'$ contains a bichromatic $6$-cycle with colours $0$ and $c(z)$, a contradiction. Hence, $I$ is an independent set.

Now, for contradiction, suppose that $F$ contains a cycle $C$.
By construction, every old vertex corresponding to a vertex of $C$ is not coloured $0$ (as otherwise we would have placed it in $I$).
We observe that new vertices are not coloured~$0$, as $x$ is coloured $0$ and every new vertex has three common neighbours with $x$, at least two of which are coloured alike.
Hence, every new vertex corresponding to an edge of $C$ is not coloured $0$.
Therefore $G'$ contains a bichromatic cycle with colours $1$ and $2$, a contradiction.
So $G$ admits an $(I,F)$-partition.
\qed
\end{proof}

\noindent
We are now ready to prove Theorem~\ref{t-main1}.

\medskip
\noindent
{\bf  Theorem~\ref{t-main1} (restated).}
{\it For $d\geq 1$ and $k\geq 3$,  
{\sc Acyclic $k$-Colouring} on graphs of diameter at most~$d$ is
polynomial-time solvable if $d=1$, $k\geq 4$ or $d\leq 2$, $k=3$ and \NP-complete if $d\geq 2$, $k\geq 4$ or  $d\geq 4$, $k=3$.}

\begin{proof}
The cases $d\leq 2$, $k=3$ and $d\geq 4$, $k=3$ follow from Lemma~\ref{l-ad2} and Theorem~\ref{lem:hard-acycl}, respectively. 
The case $d=1$, $k\geq 4$ is trivial. 
For the case $d\geq 2$, $k\geq 4$ we reduce from {\sc Acyclic $3$-Colouring}: to an instance $G$ of {\sc Acyclic $k$-Colouring}, we add a clique of $k-3$ vertices, which we make adjacent to every vertex of $G$. 
\end{proof}

\section{The Proof of Theorem~\ref{t-main2}}\label{s-star}

A {\it list assignment} of a graph $G=(V,E)$ is a function $L$ that gives each vertex $u\in V$ a {\it list of admissible colours} $L(u)\subseteq \{1,2,\ldots\}$.
A colouring $c$  {\it respects} ${L}$ if  $c(u)\in L(u)$ for every $u\in V.$ If $|L(u)|\leq 2$ for each $u\in V$, then $L$ is 
a {\it $2$-list assignment}. The {\sc $2$-List Colouring} problem is to decide if a graph $G$ with a $2$-list assignment $L$ has a colouring that respects $L$. We need the following well-known result of Edwards.

\begin{theorem}[\cite{Ed86}]\label{t-2sat}
The {\sc $2$-List Colouring} problem is solvable in time $O(n+m)$ on graphs with $n$ vertices and $m$ edges.
\end{theorem}

\noindent
We will use Theorem~\ref{t-2sat} in the proof of Lemma~\ref{l-sd2}, which is the main result of the section.  In order to do this, 
we must first be able to modify an instance of {\sc Star $3$-Colouring} into an equivalent instance of {\sc $3$-Colouring}. We can do this as follows. Let $G=(V,E)$ be a graph. We construct a supergraph $G_s$ of $G$ as follows. For each edge $e=uv$ of $G$ we add a vertex~$z_{uv}$ that we make adjacent to both $u$ and $v$. We also add an edge between two vertices $z_{uv}$ and $z_{u'v'}$ if and only if $u,v,u',v'$ are four distinct vertices such that $G$ has at least one edge with one end-vertex in $\{u,v\}$ and the other one in $\{u',v'\}$. We say that $G_s$ is the {\it edge-extension} of $G$. Observe that we constructed $G_s$ in $O(m^2)$ time. It is readily seen that $G$ has a star $3$-colouring if and only if $G_s$ has a $3$-colouring.

Now suppose that $G$ has a $2$-list assignment $L$
with $L(u)\subseteq \{1,2,3\}$ for every $u\in V$.
We extend $L$ to a list assignment $L_s$ of~$G_s$. We first set $L_s(u)=L(u)$ for every $u\in V(G)$. Initially, we set $L_s(z_e)=\{1,2,3\}$ for each edge $e\in E(G)$. We now adjust a list $L_s(z_e)$ as follows. Let $e=uv$. If $L(u)=L(v)$ or $L(u)$ has size~$1$, then we set $L_s(z_{uv})=\{1,2,3\}\setminus L(u)$.
If $L(v)$ has size~$1$, then we set $L_s(z_{uv})=\{1,2,3\}\setminus L(v)$.
If $z_{u'v'}$ is adjacent to a vertex $z_{uv}$ with $|L'(z_{uv})|=1$, then we set $L_s(z_{u'v'})=\{1,2,3\}\setminus L'(z_{uv})$. We apply the rules exhaustively. We call the resulting list assignment $L_s$ of $G_s$ the {\it edge-extension} of $L$. 
We say that an edge $uv$ of $G$ is {\it unsuitable} if $|L(u)|=|L(v)|=2$ but $L(u)\neq L(v)$, whereas $uv$ is {\it list-reducing} if $|L(u)|=|L(v)|=1$ and $L(u)\neq L(v)$. Note that in $G_s$, we may have $|L_s(z_e)|=3$ if $e$ is unsuitable, whereas $|L_s(z_e)|=1$ if $e$ is list-reducing. We say that an end-vertex $u$ of an unsuitable edge $e$ is a {\it fixer} for $e$ if $u$ is adjacent to an end-vertex of a list-reducing edge $u'v'$ (note that $\{u,v\}\cap \{u',v'\}=\emptyset$).
We make the following straightforward observation.

\begin{lemma}\label{l-ex}
Let $G=(V,E)$ be a graph on $m$ edges with a $2$-list assignment $L$ 
such that $L(u)\subseteq \{1,2,3\}$ for every $u\in V$. 
Then we can construct in $O(m^2)$ time the edge-extension~$G_s$ of $G$ and the edge-extension~$L_s$ of $L$. Moreover, $G$ has a star $3$-colouring that respects $L$ if and only if $G_s$ has a $3$-colouring that respects $L_s$. Furthermore, $L_s$ is a $2$-list assignment of~$G_s$ if  every unsuitable edge $uv$ of $G$ has a fixer. 
\end{lemma}

\noindent
Let $d_G(u)$ denote the degree of a vertex $u$ in $G$. 
We need two structural lemmas for the correctness proof of our algorithm for {\sc Star $3$-Colouring} on graphs of diameter at most~$3$.

\begin{lemma}\label{l-claim1}
Let $G$ be a graph of diameter at most~$3$.
If $G$ has a star $3$-colouring, then 
\begin{enumerate}
\item for every $4$-cycle $v_0v_1v_2v_3v_0$ of $G$, $d_G(v_0)=d_G(v_2)=2$ or $d_G(v_1)=d_G(v_3)=2$, and
\item there is no $5$-cycle in $G$.
\end{enumerate}
\end{lemma}

\begin{proof}
Assume that $G$ has a star $3$-colouring~$c$.
We first show Property~1. For contradiction, let $C$ be a (not necessarily induced) $4$-cycle $v_0v_1v_2v_3v_0$ for which the property does not hold. Then, without loss of generality, $v_0$ and $v_1$ each have degree at least~$3$ in $G$. By the pigeonhole principle, two vertices of $C$ have the same colour. As $c$ is a star colouring, we may assume without loss of generality that $c(v_0)=c(v_2)=1$ and $c(v_1)=2$ and $c(v_3)=3$. 
As $v_0$ and $v_2$ are coloured alike, $v_0$ and $v_2$ are not adjacent. As $v_0$ has degree at least~$3$, this means that $v_0$ has a neighbour $u$ that is not on $C$. As $c(v_0)=1$, we find that $c(u)\in \{2,3\}$. If $c(u)=2$, then the vertices $u$, $v_0$, $v_1$, $v_2$ form a bichromatic $P_4$. 
 If $c(u)=3$, then the vertices $u$, $v_0$, $v_3$, $v_2$ form a bichromatic $P_4$. Hence, in both cases, we obtain a contradiction.

We now show Property 2. Let $C$ be a (not necessarily induced) $5$-cycle $v_0v_1v_2v_3v_4v_0$.  Then, without loss of generality, $c(v_0)=c(v_2)=1$ and $c(v_1)=2$. As $c$ is a star $3$-colouring, $c(v_3)\neq 2$ and
$c(v_4)\neq 2$. So, $c(v_3)=c(v_4)=3$, a contradiction, as $v_3$ and $v_4$ are adjacent.\qed
\end{proof}

\begin{lemma}\label{l-claim2}
Let $G$ be a graph of diameter at most~$3$ that has two vertices $u$ and $v$ with at least three common neighbours. 
Let $w\in N(u)\cap N(v)$. Then $G$ has a star $3$-colouring if and only if $G-w$ has a star $3$-colouring. Moreover, $G-w$ has diameter at most~$3$ as well.
\end{lemma}

\begin{proof}
Let $w'\in (N(u)\cap N(v))\setminus \{w\}$. Then $uwvw'u$ is a $4$-cycle. As $u$ and $v$ have degree at least~$3$, Lemma~\ref{l-claim1} tells us that $w$, $w'$ and any other common neighbour of $u$ and $v$ have degree~$2$ in $G$. Hence, $u$ and $v$ are the only two neighbours of every vertex in $N(u)\cap N(v)$. Consequently, $G-w$ has diameter at most~$3$. 

If $G$ has a star $3$-colouring, then $G-w$ has a star $3$-colouring. Now suppose that $G-w$ has a star $3$-colouring~$c$. We claim that $c$ colours all vertices of $(N(u)\cap N(v))\setminus \{w\}$ with the same colour. Then, as $|(N(u)\cap N(v))\setminus \{w\}|\geq 2$, we can safely assign this colour to $w$ as well and obtain a star $3$-colouring of $G$. 

 If $V(G)=(N(u)\cap N(v)) \cup \{u,v\}$, then the above claim is readily seen.
Otherwise, as $G$ is connected, there exists a vertex~$z$ that is adjacent to exactly one of $u,v$, say $z$ is adjacent to $u$ but not to $v$.
For contradiction, assume that $c(w')=1$ and $c(w'')=2$, where $w''\notin \{w,w'\}$ is another vertex of $N(u)\cap N(v)$. Then $c(u)=c(v)=3$. If $c(z)=1$, then the vertices $z$, $u$, $w'$, $v$ form a bichromatic $P_4$. 
 If $c(z)=2$, then the vertices $z$, $u$, $w''$, $v$ form a bichromatic $P_4$. Hence, in both cases, we obtain a contradiction.\qed
\end{proof}

\noindent
Two non-adjacent vertices in a graph $G$ that have the same neighbourhood are {\it false twins} of $G$. 
Our algorithm for {\sc Star List $3$-Colouring} in Lemma~\ref{l-sd2} that takes as input a graph $G$ of diameter at most~$3$ can be summarized as follows:

\medskip
\noindent
{\bf Outline:}

\begin{enumerate}
\item We modify $G$ into a graph $G'$ by removing all but at most two vertices from any set of false twins of degree~2 in $O(n^2)$ time; we prove that $G$ has a star $3$-colouring if and only if $G'$ has a star $3$-colouring.
\item We then construct $O(n)$ $2$-list assignments $L'$ of $G'$, each of which in $O(n+m)$ time, such that $G'$ has a star $3$-colouring if and only if $G'$ has a star $3$-colouring respecting at least one of the constructed $2$-list assignments~$L'$. 
\item For each $(G',L')$, we prove that the edge-extension $L'_s$ of $L'$ is a $2$-list assignment of the edge-extension $G_s'$ of $G'$. Hence, it remains to solve {\sc $2$-List-Colouring} for each of the $O(n)$ instances $(G'_s,L_s')$, which we can do in $O(m^2)$ time by Theorem~\ref{t-2sat} as the size of $G_s'$ is $O(m^2)$.
\end{enumerate}

\noindent
We are now ready to give our algorithm in detail.
\begin{lemma}\label{l-sd2}
{\sc Star $3$-Colouring} is polynomial-time solvable for graphs of diameter at most~$3$.
\end{lemma}

\begin{proof}
Let $G$ be a graph of diameter~$3$. 
We may assume without loss of generality that $G$ is connected.

We first determine in $O(nm^2)$ time all $4$-cycles and all $5$-cycles in $G$. If $G$ has a $4$-cycle with two adjacent vertices of degree at least~$3$ in $G$ or if $G$ has a $5$-cycle, then $G$ is not star $3$-colourable by Lemma~\ref{l-claim1}.  We continue by assuming that $G$ satisfies the two properties of Lemma~\ref{l-claim1}. We reduce $G$ by applying Lemma~\ref{l-claim2} exhaustively. Let $G'$ be the resulting graph, which has diameter at most~$3$ (by Lemma~\ref{l-claim2}). We can determine in $O(n)$ time all vertices of degree~$2$ in $G$. For each vertex of degree~$2$ we can compute in $O(n)$ time all its false twins. Hence, we found $G'$ in $O(n^2)$ time. As we only removed vertices, $G'$ also satisfies the two properties of Lemma~\ref{l-claim1}.

If $G'$ has maximum degree at most~$4$, then $|V(G')|\leq 53$, as $G'$ has diameter at most~$3$. We can check in constant time if $|V(G')|\leq 53$ and if so,  in constant time, if $G'$ has a star $3$-colouring. Otherwise, we found a vertex $v$ of degree at least~$5$ in $G'$. 

We let $N_i$ be the set of vertices of distance~$i$ from $v$. Note that $N_1=N(v)$ and $V(G')=\{v\}\cup N_1\cup N_2\cup N_3$, as $G'$ has diameter at most~$3$.
We assume without loss of generality that if $G'$ has a star $3$-colouring, then $v$ will be coloured~$1$. We will now detect in polynomial time whether $G'$ has a star $3$-colouring $c$ with $c(v)=1$, such that 
exactly one of the following situations hold: $c$ gives each vertex in $N_1$ colour $3$; or $c$ gives at least one vertex of $N_1$ colour~$2$ and at least three vertices of $N_1$ colour~$3$.
As $v$ has degree at least~$5$, at least one of colours $2$, $3$ must occur three times on $N(v)$, and we may assume without loss of generality that this colour is $3$. Hence, $G'$ has a star $3$-colouring if and only if one of these two situations holds.

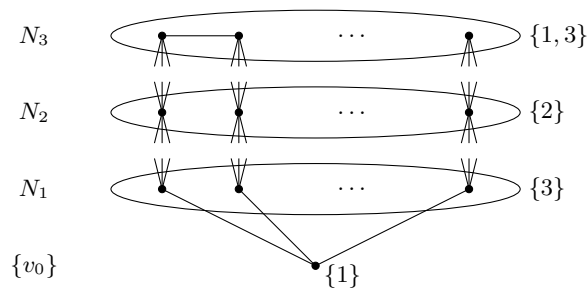
\begin{figure}
\begin{center}
\begin{tikzpicture}[scale=0.68]

%N_0
\draw[fill=black] (0,0) circle[radius=2pt];
\node[right] at (0,-0.2){$\{1\}$};
%N_1
\draw[fill=black] (-3,1.5) circle[radius=2pt];
\draw[fill=black] (-1.5,1.5) circle[radius=2pt];
\draw (0.75,1.5) node{\ldots};
\draw[fill=black] (3,1.5) circle[radius=2pt];
\draw (0,1.5) ellipse (4 and 0.5);
\node[right] at (4,1.5){$\{3\}$};
\draw (0,0)--(-3,1.5);
\draw (0,0)--(-1.5,1.5);
%\draw (0,0)--(0.45,0.6);
%\draw (0,0)--(0.3,0.6);
\draw (0,0)--(3,1.5);
\draw (-3,1.5)-- (-2.85,2.1);
\draw (-3,1.5)-- (-3,2.1);
\draw (-3,1.5)-- (-3.15,2.1);
\draw (3,1.5)-- (2.85,2.1);
\draw (3,1.5)-- (3,2.1);
\draw (3,1.5)-- (3.15,2.1);
\draw (-1.5,1.5)-- (-1.65,2.1);
\draw (-1.5,1.5)-- (-1.5,2.1);
\draw (-1.5,1.5)-- (-1.35,2.1);
%N_2
\draw[fill=black] (-3,3) circle[radius=2pt];
\draw[fill=black] (-1.5,3) circle[radius=2pt];
\draw (0.75,3) node{\ldots};
\draw[fill=black] (3,3) circle[radius=2pt];
\draw (0,3) ellipse (4 and 0.5);
\node[right] at (4,3){$\{2\}$};
\draw (-3,3)-- (-2.85,2.4);
\draw (-3,3)-- (-3.15,2.4);
\draw (-3,3)-- (-3,2.4);
\draw (3,3)-- (2.85,2.4);
\draw (3,3)-- (3.15,2.4);
\draw (3,3)-- (3.,2.4);
\draw (-1.5,3)-- (-1.65,2.4);
\draw (-1.5,3)-- (-1.35,2.4);
\draw (-1.5,3)-- (-1.5,2.4);
\draw (-3,3)-- (-2.85,3.6);
\draw (-3,3)-- (-3.15,3.6);
\draw (-3,3)-- (-3,3.6);
\draw (3,3)-- (2.85,3.6);
\draw (3,3)-- (3.15,3.6);
\draw (3,3)-- (3,3.6);
\draw (-1.5,3)-- (-1.65,3.6);
\draw (-1.5,3)-- (-1.35,3.6);
\draw (-1.5,3)-- (-1.5,3.6);
%N_3
\draw[fill=black] (-3,4.5) circle[radius=2pt];
\draw[fill=black] (-1.5,4.5) circle[radius=2pt];
\draw (0.75,4.5) node{\ldots};
\draw[fill=black] (3,4.5) circle[radius=2pt];
\draw (0,4.5) ellipse (4 and 0.5);
\node[right] at (4,4.5){$\{1,3\}$};
\draw (-3,4.5)--(-1.5,4.5);
\draw (-3,4.5)-- (-2.85,3.9);
\draw (-3,4.5)-- (-3.15,3.9);
\draw (-3,4.5)-- (-3,3.9);
\draw (3,4.5)-- (2.85,3.9);
\draw (3,4.5)-- (3.15,3.9);
\draw (3,4.5)-- (3,3.9);
\draw (-1.5,4.5)-- (-1.65,3.9);
\draw (-1.5,4.5)-- (-1.35,3.9);
\draw (-1.5,4.5)-- (-1.5,3.9);
% Left side
\draw (-5.5,0) node{$\{v_0\}$};
\draw (-5.5,1.5) node{$N_1$};
\draw (-5.5,3) node{$N_2$};
\draw (-5.5,4.5) node{$N_3$};
\end{tikzpicture}
\end{center}
\caption{The pair $(G',L')$ in Case 1.}\label{f-case1}
\end{figure}

\medskip
\noindent
{\bf Case 1.} Check if $G'$ has a star $3$-colouring that gives every vertex of $N_1$ colour~$3$.

\smallskip
\noindent
As $|N_1|\geq 5$, such a star $3$-colouring $c$ must assign each vertex of $N_2$ colour~$2$. This means that every vertex of $N_3$ gets colour $1$ or $3$. Hence, we obtained, in $O(n)$ time, a $2$-list assignment~$L'$ of~$G'$. We construct the pair $(G_s',L_s')$. By Lemma~\ref{l-ex} this take $O(m^2)$ time. As every list either has size~1 or is equal to $\{1,3\}$, we find that the edge-extension $L_s'$ of $L'$ is a $2$-list assignment of $G_s'$.
By Lemma~\ref{l-ex}, it remains to solve {\sc $2$-List-Colouring} on $(G'_s,L_s')$. We can do this in $O(m^2)$ time using Theorem~\ref{t-2sat} as the size of $G'_s$ is $O(m^2)$. Hence, the total running time for dealing with Case~1 is $O(m^2)$. See also Figure~\ref{f-case1}.

\medskip
\noindent
{\bf Case 2.} Check if $G'$ has a star $3$-colouring that gives at least one vertex of $N_1$ colour~$2$ and at least three vertices of $N_1$ colour~$3$.

\smallskip
\noindent
We set $L'(v)=\{1\}$. 
This gives us the following property.

\begin{enumerate}
\item [{\bf P0.}] $N_0=\{v\}$ and $L'(v)=\{1\}$.
\end{enumerate}

We now select four arbitrary vertices of $N(v)$. We consider all possible colourings of these four vertices with colours $2$ and $3$, where we assume without loss of generality that colour $3$ is used on these four vertices at least as many times as colour~$2$. For the case where colour~$2$ is not used we consider each of the $O(n)$ options of colouring another vertex from $N(v)$ with colour~$2$. For the cases where colour~$3$ is used exactly twice, we consider each of the $O(n)$ options of colouring another vertex from $N(v)$ with colour~$3$. Hence, the total number of options is $O(n)$, and in each option we have a neighbour $x$ of $v$ with colour~$2$ and a set $$W=\{w_1,w_2,w_3\}$$ of three distinct neighbours of $v$ with colour~$3$. That is, we set $L'(x)=\{2\}$ and $L'(w_i)=\{3\}$ for $1\leq i\leq 3$. 

For each set $\{x\}\cup W$ we do as follows. We first check if $W$ is independent; otherwise we discard the option. If $W$ is independent, then initially we set $L'(u)=\{1,2,3\}$ for each $u\notin \{x,v\}\cup W$. We now show that we can reduce the list of every such vertex $u$ by at least~$1$. 
As an {\it implicit step}, we will discard the instance $(G',L')$ if one of the lists has become empty. 
In doing this we will use the following {\it Propagation Rule}: 

\medskip
\noindent
{\it Whenever a vertex has only one colour in its list, we remove that colour from the list of each of its neighbours.}
 
\medskip
\noindent
By the Propagation Rule, we obtain the following property, in which we updated the set $W$:

\begin{enumerate}
\item [{\bf P1.}] $N_1$ can be partitioned into sets $W,X,Y$ with  $|W|\geq 3$, $|X|\geq 1$ and $|Y|\geq 0$, such that no vertex of $Y$ is adjacent to any vertex of $X\cup W$, and moreover, $X$ is an independent set with $x\in X$ and $W$ is an independent set with $\{w_1,w_2,w_3\}\subseteq W$, such that
\begin{itemize}
\item every vertex $w\in W$ has list $L'(w)=\{3\}$,
\item every vertex $x\in X$ has list $L'(x)=\{2\}$, and 
\item every vertex $y\in Y$ has list $L'(y)=\{2,3\}$. 
\end{itemize}
\end{enumerate}

\noindent
Note that by the Propagation Rule, we removed colour~$3$ from the list of every neighbour of a vertex of $W$ in $N_2$.
We now also remove colour~$1$ from the list of every neighbour of a vertex of $W$ in $N_2$; the reason for this is that if a neighbour $y$ of, say, $w_1$ is coloured~$1$, then the vertices $y,w_1,v,w_2$ form a bichromatic~$P_4$. Hence, any neighbour of every vertex in $W$ in $N_2$ has list $\{2\}$.

Now consider a vertex $z\in N_2$ that still has a list of size~$3$. Then $z$ is not adjacent to any vertex in 
$N_1$ with a singleton list (as otherwise we applied the Propagation Rule), 
but by definition $z$ still has a neighbour $z'$ in $N_1$. This means that $z' \in Y$ and thus $z'$ has list $\{2,3\}$. Hence, $z$ cannot be coloured~$1$: if $z'$ gets colour~$2$, the vertices $x,v,z',z$ will form a bichromatic~$P_4$, and if $z'$ gets colour~$3$, the vertices $w_1,v,z',z$ will form a bichromatic~$P_4$. Hence, we may remove colour~$1$ from $L'(z)$, so $L'(z)$ will have size at most~$2$. 

We make some more observations. First, we recall that every neighbour of a vertex in $W$ in $N_2$ has list $\{2\}$, and every vertex in $X$ has list $\{2\}$ as well. Hence, no vertex in $N_2$ has both a neighbour in $W$ and a neighbour in $X$; otherwise this vertex would have an empty list by the Propagation Rule and we would have discarded this option.

Due to the above, we can partition $N_2$ into sets $W^*$, $X^*$, and $Y^*$ such that the vertices of $W^*$ are the neighbours of $W$ and the vertices of $X^*$ are the neighbours of $X$, whereas $Y^*=N_2\setminus (X^*\cup W^*)$. Consequently, the neighbours in $N_1$ of every vertex of $Y^*$ belong to~$Y$. 

Recall that $G'$ has no $5$-cycles. Hence, there is no edge between vertices from two different sets of $\{W^*,X^*,Y^*\}$. Furthermore, every vertex $w^*\in W^*$ has list $L'(w^*)=\{2\}$, every vertex $x^*\in X^*$ has list $L'(x^*)=\{1,3\}$, and every vertex $y^*\in Y^*$ has list $L'(y^*)=\{2,3\}$.
If a vertex $y\in Y$ has a neighbour $w^*\in W^*$, then $vww^*yv$ is a $4$-cycle where $w\in W$ is a neighbour of~$w^*$. 
Recall that $G'$ satisfies the properties of Lemma~\ref{l-claim1}. As $v$ has degree at least~$5$ in $G'$, this means that $y$ has degree~$2$ in $G'$. Hence, $v$ and $w^*$ are the only neighbours of~$y$. In particular, we find that every vertex in $Y$ with a neighbour in $W^*$ has no neighbour in $X^*\cup Y^*$.

We now apply the Propagation Rule again. As a consequence, we update the lists of the vertices in $Y\cup N_3$, the sets $Y$ and $W$ in {\bf P1}. The latter is because some vertices might have moved from $Y$ to $W$; 
in particular it now holds that no vertex in $W^*$ is adjacent to any vertex in $Y$.

We summarize the above in the following property:

\begin{enumerate}
\item [{\bf P2.}] $N_2$ can be partitioned into sets $W^*$, $X^*$ and $Y^*$, such that 
\begin{itemize} 
\item every vertex $w^*\in W^*$ has list $L'(w^*)=\{2\}$ and all its neighbours in $N_1$ belong to~$W$,
\item every vertex $x^*\in X^*$ has list $L'(x^*)\subseteq \{1,3\}$ and at least one of its neighbours in $N_1$ belong to $X$ and none of them belong to~$W$, 
\item every vertex $y^*\in Y^*$ has list $L'(y^*)\subseteq \{2,3\}$ and all its neighbours in $N_1$ belong to~$Y$, and
\item there is no edge between vertices from two different sets of $\{W^*,X^*,Y^*\}$.
\end{itemize}
\end{enumerate}

\noindent
We now consider the set $N_3$.
We let $T_1$ be the set consisting of all vertices in $N_3$ that have at least two neighbours in $W^*$.
We let $T_2$ be the set consisting of all vertices in $N_3$ that have exactly one neighbour in $W^*$.
Moreover, we let $S_1$ be the set of vertices of $N_3\setminus (T_1\cup T_2)$ that have
 at least one 
 neighbour in $T_1$.
We let $S_2$ be the set of vertices of $N_3\setminus (T_1\cup T_2)$ that have no neighbours in $T_1$ but
at least two neighbours in $T_2$.
If for a vertex $s\in N_3$, there is a vertex $w\in W$ and a $4$-path from $s$ to $w$ whose internal vertices are in $X$ and $X^*$, then we let $s\in R$.

We note that the sets $S_1$, $S_2$, $T_1$ and $T_2$ are pairwise disjoint by definition, whereas the set $R$ may intersect with
$S_1\cup S_2\cup T_1\cup T_2$. We now show that $N_3=R\cup S_1\cup S_2\cup T_1\cup T_2$.
For contradiction, assume that $s$ is a vertex of $N_3$ that does not belong to any of the five sets $R,S_1,S_2,T_1, T_2$.
As $s\notin T_1\cup T_2$, we find that the distance from $s$ to every vertex of $W$ is at least~$3$.
Then, as $G'$ has diameter~$3$, there exists a $4$-path $P_i$ from $s$ to each $w_i\in W$ (by {\bf P1} we can write $W^*=\{w_1,\ldots,w_a\}$ for some $a\geq 3$).
 Every $P_i$ must be of one of the following forms: $s-N_2-N_1-w_i$ or $s-N_2-N_2-w_i$ or $s-N_3-N_2-w_i$. 

First assume that there exists some $P_i$ that is of the form $s-N_2-N_1-w_i$, that is, $P_i=szz'w_i$ for some $z\in N_2$ and $z'\in N_1$. As $z'$ is a neighbour of both $w_i$ and~$v$, we find that $z'\in X$ and $z'\in X^\star$, and consequently, $s\in R$, a contradiction.

Now assume that there exists some $P_i$ that is of the form $s-N_2-N_2-w_i$, that is, $P_i=szz'w_i$ for some $z$ and $z'$ in $N_2$.
By definition, $z$ must have a neighbour in $N_1$.  
As $G'$ has no $5$-cycle, this is only possible if $z$ is adjacent to $w_i$. However, now $s$ is no longer of distance~$3$ from $w_i$ in $G'$, 
a contradiction.

Finally, assume that no path from $s$ to any $w_i$ is of one of the two forms above. Hence, every $P_i$ is of the form $s-N_3-N_2-w_i$. We write $P_i=st_iw_i^*w_i$ where $t_i\in T_1\cup T_2$ and $w_i^*\in W^*$. We consider the paths $P_1$, $P_2$, $P_3$, which exist as $|W|\geq 3$.
As $s\notin S_1$, we find that $t_i\notin T_1$.
Moreover, as $s\notin S_2$, we find that $t_1=t_2=t_3$, and so $w_1^*=w_2^*=w_3^*$. In particular, the latter implies that $w_1^*$ is adjacent to $w_1$, $w_2$ and $w_3$ and thus has degree at least~$3$.
Recall that $G'$ satisfies Property~1 of Lemma~\ref{l-claim1}. As $w_1^*$ and $v$ each have degree at least~$3$ in $G'$, this means that each $w_i$ must only be adjacent to $v$ and $w_1^*$. However, then $w_1$, $w_2$ and $w_3$ are three false twins of degree~2 in $G'$, and by construction of $G'$ we would have removed one of them, a contradiction. 
We conclude that $N_3=R\cup S_1\cup S_2\cup T_1\cup T_2$.

We now reduce the lists of the vertices in~$N_3$.
Let $s\in N_3$.
First suppose that $s\in T_1\cup T_2$, that is, $s$ is adjacent to a vertex $w^*\in W^*$.
Then, as $L'(w^*)=\{2\}$, we find that $L'(s)\subseteq \{1,3\}$. 
If $s\in T_1$, then we can reduce the list of $s$ further as follows. By the definition of $T_1$, we have that $s$ is adjacent to a second vertex $w'\neq w^*$ in $W^*$. By {\bf P2}, we find that $w'$ has a neighbour $w\in W$.
We find that $L'(w^*)=L(w')=\{2\}$ and $L(w)=\{3\}$. Then $s$ cannot be assigned colour~$3$, as otherwise $w^*,s,w',w$ would form a bichromatic~$P_4$. Hence, we can reduce the list of $s$ from $\{1,3\}$ to $\{1\}$. 

Now suppose that $s\in S_1$. Then, by the definitions of the sets $S_1$ and $T_1$ and {\bf P2}, there exists a path $P=stw^*w$ where $t\in T_1$, $w^*\in W^*$ and $w\in W$. We deduced above that $t$ has list $L'(t)=\{1\}$. Consequently, we can delete colour~$1$ from the list of $s$ by the Propagation Rule, so $L'(s)\subseteq \{2,3\}$.

Now suppose that $s\in S_2$. Then, by the definition of $S_2$ and {\bf P2}, there exist two paths $P_1=st_1w_1^*w_1$ and $P_2=st_2w_2^*w_2$ where $t_1,t_2\in T_2$, $w_1^*,w_2^*\in W^*$, $w_1,w_2\in W$, and $t_1\neq t_2$.
We claim that $s$ cannot be assigned colour~$2$. For contradiction, suppose that $s$ has colour~$2$. Then $t_1$, which has list $\{1,3\}$, must receive colour~$1$, as otherwise $t_1$ will have colour~$3$ and $s,t_1,w_1^*,w_1$ is a bichromatic~$P_4$ (recall that $w_1^*$ and $w_1$ can only be coloured with colours $2$ and $3$, respectively). For the same reason, $t_2$ must get colour~$1$ as well. However, now $w_1^*,t_1,s,t_2$ is a bichromatic~$P_4$, a contradiction. Hence, we can remove colour~$2$ from $L'(s)$. 
Afterwards, $L'(s)\subseteq \{1,3\}$.  

Finally, suppose that $s\in R$. By the definition of $R$, there is some path $P_i=sx^*x'w$ where $x^*\in X^*$, $x'\in X$, and $w\in W$. 
By {\bf P1} and {\bf P2}, respectively, it holds that $L'(x')=\{2\}$ and $L'(x^*)\subseteq \{1,3\}$. Hence, $s$ cannot be coloured~$2$: if $x^*$ gets colour~$1$, the vertices $v,x',x^*,s$ will form a bichromatic~$P_4$, and if $x^*$ gets colour~$3$, the vertices $w_1,x',x^*,s$ will form a bichromatic~$P_4$. In other words, we may remove colour~$2$ from $L'(s)$, so $L'(s)\subseteq \{1,3\}$.

As $N_3=R\cup S_1\cup S_2\cup T_1\cup T_2$, we found that every vertex of $N_3$ has a list of size at most~$2$ and more specifically
we obtained the following property:

\begin{enumerate}
\item [{\bf P3.}] $N_3$ only consists of vertices whose lists are a subset of $\{1,3\}$ or $\{2,3\}$, and $N_3$ can be split into sets $R,S_1,S_2,T_1, T_2$, such that $S_1$, $S_2$, $T_1$ and $T_2$ are pairwise disjoint, and
\begin{itemize}
\item every vertex $r\in R$ has list $L'(r)\subseteq \{1,3\}$ and there is a $4$-path from~$r$ to a vertex in $W$ that has its two internal vertices in $X^*$ and $X$, respectively,
\item every vertex $t\in T_1$ has list $L'(t)=\{1\}$ and has at least two neighbours in $W^*$,
\item every vertex $t\in T_2$ has list $L'(t)\subseteq \{1,3\}$ and has exactly one neighbour in $W^*$,
\item every vertex $s\in S_1$ has list $L'(s)\subseteq \{2,3\}$, has no neighbours in $W^*$ but is adjacent to
 at least one vertex in $T_1$, and
\item every vertex $s\in S_2$ has list $L'(s)\subseteq \{1,3\}$ and has no neighbours in $T_1\cup W^*$ but 
at least 
two neighbours in $T_2$. 
\end{itemize}
\end{enumerate}

\noindent
We conclude that we constructed a set $\mathcal{L}'$ of $2$-list assignments of $G'$, such that $\mathcal{L}'$ is of size $O(n)$ and $G'$ has a star $3$-colouring if and only if $G'$ has a star $3$-colouring that respects $L'$ for some $L'\in\mathcal{L}'$. 
Moreover, we can find each $L'\in {\mathcal L}$ in $O(m+n)$ time by a breadth-first search for detecting the $4$-paths.
For each $L'\in {\mathcal L}$, we do as follows.

We still need to construct the edge-extension $G_s'$ of $G'$. However, the edge-extension  $L'_s$ of $L'$ might not be a $2$-list assignment. The reason is that $G'$ may have an edge $ss'$ for some vertex $s\in N_2$ with $L'(s)=\{2,3\}$ and some vertex $s'\in N_3$ with $L'(s')=\{1,3\}$ such that $L'_s(z_{ss'})=\{1,2,3\}$. 
We distinguish between two cases. See Figure~\ref{f-2a} for the situation in Case~2a and Figure~\ref{f-2b} for the situation of Case~2b.

\begin{figure}[t]
\begin{center}
\begin{tikzpicture}[scale=0.68]
%N_0
\draw[fill=black] (0,0) circle[radius=2pt];
\node[right] at (0,-0.2){$\{1\}$};
%N_1
\draw (-3.75,1.5) ellipse (1.25 and 0.5);
\draw[fill=black] (-4.7,1.5) circle[radius=2pt];
\draw[fill=black] (-4.2,1.5) circle[radius=2pt];
\draw (-3.5,1.5) node{\ldots};
\draw[fill=black] (-2.8,1.5) circle[radius=2pt];
\draw (-4.6,1.3) node[below left]{$W$}; 
\draw (-2.5,1.5) node[right]{$\{3\}$}; 
\draw (-4.7,1.5)--(0,0); 
\draw (-4.2,1.5)--(0,0); 
\draw (-2.8,1.5)--(0,0);
\draw (-4.7,1.5)--(-4.85,2.1);
\draw (-4.7,1.5)--(-4.7,2.1);
\draw (-4.7,1.5)--(-4.55,2.1); 
\draw (-4.2,1.5)--(-4.35,2.1);
\draw (-4.2,1.5)--(-4.2,2.1);
\draw (-4.2,1.5)--(-4.05,2.1); 
\draw (-2.8,1.5)--(-3.1,2.1);
\draw (-2.8,1.5)--(-2.8,2.1);
\draw[fill=black] (0,1.5) circle[radius=2pt];
\draw (0,1.5) ellipse (0.8 and 0.5);
\draw (0.8,1.5) node[right] {$\{2\}$};
\draw (-0.4,1.3) node[below left]{$X$}; 
\draw (0,1.5)--(0,0); 
\draw (0,1.5)--(-0.15,2.1);
\draw (0,1.5)--(0.15,2.1); 
\draw (5,1.5) node[right] {$\{\not 2,3\}$};
\draw (4.6,1.3) node[below right] {$Y$};
\draw[fill=black] (4,1.5) circle[radius=2pt];
\draw (4,1.5) ellipse (1 and 0.5); 
\draw (4,1.5)--(0,0);
\draw (4.15,2.1)--(4,1.5);
\draw (3.85,2.1)--(4,1.5);
%N_2
\draw (-3.75,3) ellipse (1.25 and 0.5);
\draw[fill=black] (-4.2,3) circle[radius=2pt];
\draw (-2.5,3) node[right]{$\{2\}$}; 
\draw (-4.5,2.8) node[below left]{$W^\star$}; 
\draw (-4.2,3)--(-4.2,1.5); 
\draw (-4.2,3)--(-4.35,2.4);
\draw (-4.2,3)--(-4.05,2.4); 
\draw (-4.2,3)--(-4.2,3.6); 
\draw (-4.2,3)--(-4.35,3.6); 
\draw (-3.3,3)--(-3.3,2.4); 
\draw (-3.3,3)--(-3.45,2.4);
\draw (-3.3,3)--(-3.3,3.6); 
\draw (-3.3,3)--(-3.15,3.6); 
\draw (-3.3,3)--(-3.45,3.6); 
\draw (0,3) ellipse (1 and 0.5);
\draw (-0.5,2.8) node[below left]{$X^\star$}; 
\draw (1,3) node[right] {$\{1,3\}$};
\draw (0,3)--(0,0); 
%\draw (0,3)--(1.2,2.4); 
%\draw (0,3)--(1.45,2.4); 
%\draw (0,3)--(1.7,2.4); 
\draw (0,3)--(-.15,3.6); 
\draw (0,3)--(0,3.6); 
\draw[fill=black] (0,3) circle[radius=2pt];
\draw (0,3)--(-0.15,2.4);
\draw (0,3)--(0.15,2.4);  
\draw (0,3)--(0,3.6);  
\draw (0,3)--(-0.15,3.6);  
\draw (5,3) node[right] {$\{2,\not 3\}$};
\draw (4.6,2.8) node[below right] {$Y^\star$};
\draw[fill=black] (4,3) circle[radius=2pt];
\draw (4, 3) ellipse (1 and 0.5); 
\draw (4,1.5)--(4,3);
\draw (3.85,2.4)--(4,3);
\draw (4.15,2.4)--(4,3);
\draw (4,3.6)--(4,3);
\draw (3.85,3.6)--(4,3);
\draw (4.15,3.6)--(4,3);
%N_3
%
\draw (-6,5) ellipse (1 and 1);
\draw (-6.4,4.4) node[below left]{$T_1$}; 
\draw (-5.1,4.9) node[right] {$\{1\}$};
\draw[fill=black] (-6,5) circle[radius=2pt];
\draw[fill=black] (-3.3,3) circle[radius=2pt];
\draw (-4.2,3)--(-6,5)--(-3.3,3)--(-2.8,1.5);
\draw (-2,5) ellipse (1 and 1);
\draw (-2.4,4.4) node[below left]{$T_2$}; 
\draw (-1.1,4.9) node[right] {$\{1,3\}$};
\draw[fill=black] (-1.4,5) circle[radius=2pt];
\draw (-1.4,5)--(-3.3,3);
\draw[fill=black] (-2.6,5.6) circle[radius=2pt];
\draw (-2.6,5.6)--(-4.2,3);
%\draw[fill=black] (-2.6,5) circle[radius=2pt];
%\draw (-2.6,5)--(-4.2,3);
%
\draw (2,5) ellipse (1 and 1);
\draw (1.6,4.4) node[below left]{$S_1$}; 
\draw (2.9,4.9) node[right] {$\{2,3\}$};
\draw[fill=black] (2,5) circle[radius=2pt];
\draw[out=135, in=45] (2,5)to(-6,5);
\draw (6,5) ellipse (1 and 1);
\draw (5.6,4.4) node[below left]{$S_2$}; 
\draw (6.9,4.9) node[right] {$\{1,3\}$};
\draw[fill=black] (6.4,5) circle[radius=2pt];
%\draw[in=135, out=45] (-2.6,5)to(6.4,5);
\draw[in=135, out=45] (-1.4,5)to(6.4,5);
\draw[in=135, out=45] (-2.6,5.6)to(6.4,5);
\draw (-6.6,5.4) rectangle (6.6, 7.5);
\draw (-2.6,5.6) --(0,3);
\draw[fill=black] (1,6.2) circle[radius=2pt];
\draw (1,6.2) --(0,3);
\draw[in=135, out=45] (-2.8,1.5)to(0,1.5);
\draw (-6.5,7.4) node[above left]{$R$}; 
\draw (6.5,7.4) node[above right]{$\{1,3\}$}; 
% Left side
\draw (-10.5,0) node{$\{v_0\}$};
\draw (-10.5,1.5) node{$N_1$};
\draw (-10.5,3) node{$N_2$};
\draw (-10.5,5.5) node{$N_3$};

\end{tikzpicture}
\end{center}
%\caption{An example of a pair $(G',L')$ in Case 2a.}\label{f-2a}
\caption{An example of a pair $(G',L')$ in Case 2a. The colours crossed out show the difference between the general situation in Case~2 and what we show holds in Case~2a.}\label{f-2a}
\end{figure}
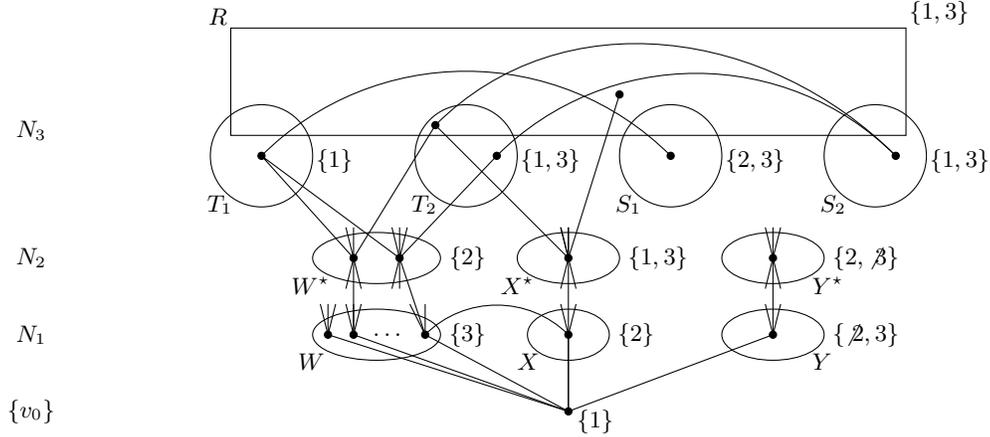
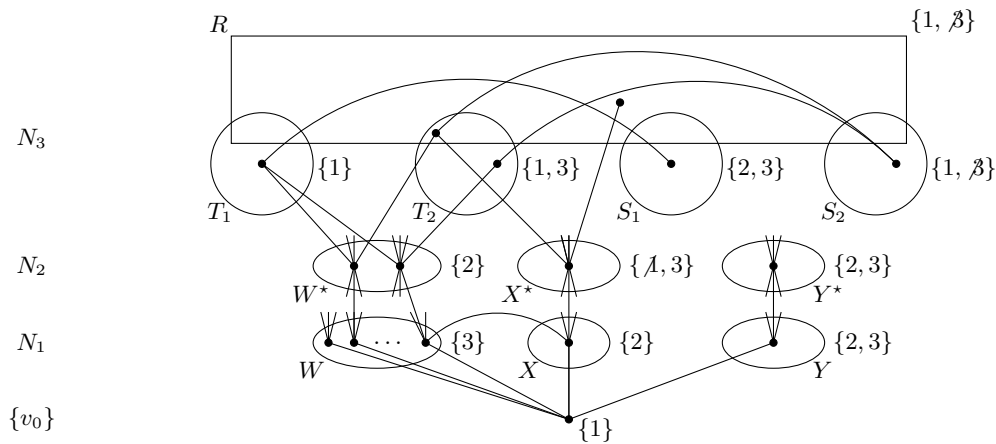
\begin{figure}[t]
\begin{center}
\begin{tikzpicture}[scale=0.68]

%N_0
\draw[fill=black] (0,0) circle[radius=2pt];
\node[right] at (0,-0.2){$\{1\}$};
%N_1
\draw (-3.75,1.5) ellipse (1.25 and 0.5);
\draw[fill=black] (-4.7,1.5) circle[radius=2pt];
\draw[fill=black] (-4.2,1.5) circle[radius=2pt];
\draw (-3.5,1.5) node{\ldots};
\draw[fill=black] (-2.8,1.5) circle[radius=2pt];
\draw (-4.6,1.3) node[below left]{$W$}; 
\draw (-2.5,1.5) node[right]{$\{3\}$}; 
\draw (-4.7,1.5)--(0,0); 
\draw (-4.2,1.5)--(0,0); 
\draw (-2.8,1.5)--(0,0);
\draw (-4.7,1.5)--(-4.85,2.1);
\draw (-4.7,1.5)--(-4.7,2.1);
\draw (-4.7,1.5)--(-4.55,2.1); 
\draw (-4.2,1.5)--(-4.35,2.1);
\draw (-4.2,1.5)--(-4.2,2.1);
\draw (-4.2,1.5)--(-4.05,2.1); 
\draw (-2.8,1.5)--(-3.1,2.1);
\draw (-2.8,1.5)--(-2.8,2.1);
\draw[fill=black] (0,1.5) circle[radius=2pt];
\draw (0,1.5) ellipse (0.8 and 0.5);
\draw (0.8,1.5) node[right] {$\{2\}$};
\draw (-0.4,1.3) node[below left]{$X$}; 
\draw (0,1.5)--(0,0); 
\draw (0,1.5)--(-0.15,2.1);
\draw (0,1.5)--(0.15,2.1); 
\draw (5,1.5) node[right] {$\{2,3\}$};
\draw (4.6,1.3) node[below right] {$Y$};
\draw[fill=black] (4,1.5) circle[radius=2pt];
\draw (4,1.5) ellipse (1 and 0.5); 
\draw (4,1.5)--(0,0);
\draw (4.15,2.1)--(4,1.5);
\draw (3.85,2.1)--(4,1.5);
%N_2
\draw (-3.75,3) ellipse (1.25 and 0.5);
\draw[fill=black] (-4.2,3) circle[radius=2pt];
\draw (-2.5,3) node[right]{$\{2\}$}; 
\draw (-4.5,2.8) node[below left]{$W^\star$}; 
\draw (-4.2,3)--(-4.2,1.5); 
\draw (-4.2,3)--(-4.35,2.4);
\draw (-4.2,3)--(-4.05,2.4); 
\draw (-4.2,3)--(-4.2,3.6); 
\draw (-4.2,3)--(-4.35,3.6); 
\draw (-3.3,3)--(-3.3,2.4); 
\draw (-3.3,3)--(-3.45,2.4);
\draw (-3.3,3)--(-3.3,3.6); 
\draw (-3.3,3)--(-3.15,3.6); 
\draw (-3.3,3)--(-3.45,3.6); 
\draw (0,3) ellipse (1 and 0.5);
\draw (-0.5,2.8) node[below left]{$X^\star$}; 
\draw (1,3) node[right] {$\{\not 1,3\}$};
\draw (0,3)--(0,0); 
%\draw (0,3)--(1.2,2.4); 
%\draw (0,3)--(1.45,2.4); 
%\draw (0,3)--(1.7,2.4); 
\draw (0,3)--(-.15,3.6); 
\draw (0,3)--(0,3.6); 
\draw[fill=black] (0,3) circle[radius=2pt];
\draw (0,3)--(-0.15,2.4);
\draw (0,3)--(0.15,2.4);  
\draw (0,3)--(0,3.6);  
\draw (0,3)--(-0.15,3.6);  
\draw (5,3) node[right] {$\{2,3\}$};
\draw (4.6,2.8) node[below right] {$Y^\star$};
\draw[fill=black] (4,3) circle[radius=2pt];
\draw (4, 3) ellipse (1 and 0.5); 
\draw (4,1.5)--(4,3);
\draw (3.85,2.4)--(4,3);
\draw (4.15,2.4)--(4,3);
\draw (4,3.6)--(4,3);
\draw (3.85,3.6)--(4,3);
\draw (4.15,3.6)--(4,3);
%N_3
%
\draw (-6,5) ellipse (1 and 1);
\draw (-6.4,4.4) node[below left]{$T_1$}; 
\draw (-5.1,4.9) node[right] {$\{1\}$};
\draw[fill=black] (-6,5) circle[radius=2pt];
\draw[fill=black] (-3.3,3) circle[radius=2pt];
\draw (-4.2,3)--(-6,5)--(-3.3,3)--(-2.8,1.5);
\draw (-2,5) ellipse (1 and 1);
\draw (-2.4,4.4) node[below left]{$T_2$}; 
\draw (-1.1,4.9) node[right] {$\{1,3\}$};
\draw[fill=black] (-1.4,5) circle[radius=2pt];
\draw (-1.4,5)--(-3.3,3);
\draw[fill=black] (-2.6,5.6) circle[radius=2pt];
\draw (-2.6,5.6)--(-4.2,3);
%\draw[fill=black] (-2.6,5) circle[radius=2pt];
%\draw (-2.6,5)--(-4.2,3);
%
\draw (2,5) ellipse (1 and 1);
\draw (1.6,4.4) node[below left]{$S_1$}; 
\draw (2.9,4.9) node[right] {$\{2,3\}$};
\draw[fill=black] (2,5) circle[radius=2pt];
\draw[out=135, in=45] (2,5)to(-6,5);
\draw (6,5) ellipse (1 and 1);
\draw (5.6,4.4) node[below left]{$S_2$}; 
\draw (6.9,4.9) node[right] {$\{1,\not 3\}$};
\draw[fill=black] (6.4,5) circle[radius=2pt];
%\draw[in=135, out=45] (-2.6,5)to(6.4,5);
\draw[in=135, out=45] (-1.4,5)to(6.4,5);
\draw[in=135, out=45] (-2.6,5.6)to(6.4,5);
\draw (-6.6,5.4) rectangle (6.6, 7.5);
\draw (-2.6,5.6) --(0,3);
\draw[fill=black] (1,6.2) circle[radius=2pt];
\draw (1,6.2) --(0,3);
\draw[in=135, out=45] (-2.8,1.5)to(0,1.5);
\draw (-6.5,7.4) node[above left]{$R$}; 
\draw (6.5,7.4) node[above right]{$\{1,\not 3\}$}; 
% Left side
\draw (-10.5,0) node{$\{v_0\}$};
\draw (-10.5,1.5) node{$N_1$};
\draw (-10.5,3) node{$N_2$};
\draw (-10.5,5.5) node{$N_3$};

\end{tikzpicture}
\end{center}
\caption{An example of a pair $(G',L')$ in Case 2b. The colours crossed out show the difference between the general situation in Case~2 and what we show holds in Case~2b. }\label{f-2b}
\end{figure}

\medskip
\noindent
{\bf Case 2a.} Check if $G'$ has a star $3$-colouring that gives $x$ colour~$2$ and every other vertex of $N_1$ colour~$3$.

\smallskip
\noindent
We only consider this case if $|X|=1$. We give every vertex in $Y$ list $\{3\}$. Then, by the Propagation Rule, we can delete colour~$3$ from every list of a vertex in $Y^*$. We construct $G'_s$ and $L'_s$ in $O(m^2)$ time by Lemma~\ref{l-ex}. Then $L'_s$ is a $2$-list assignment of $G'_s$. This can be seen as follows. Let $e=ss'$ be an unsuitable edge of $G'$. As $G'$ has no vertices with list $\{1,2\}$, we find that $L'(s)=\{2,3\}$ and $L'(s')=\{1,3\}$. Then $s$ must be in $S_1$. By definition, it follows that there exist vertices $t\in T_1$ and $w^*\in W^*$ such that $st$ and $tw^*$ are edges of $G'$. As $L'(t)=\{1\}$ and $L'(w^*)=\{2\}$, the edge $tw^*$ is list-reducing. Hence, $s$ is a fixer for the edge $ss'$. The claim now follows from Lemma~\ref{l-ex}, and by the same lemma, it remains to check if $G'_s$ has a $3$-colouring that respects $L'_s$. We can do the latter in $O(m^2)$ time by Theorem~\ref{t-2sat}.

\medskip
\noindent
{\bf Case 2b.} Check if $G'$ has a star $3$-colouring that gives at least one other vertex of $N_1$, besides $x$, colour~$2$.

\smallskip
\noindent
If $|X|\geq 2$, then we found a vertex of $N_1\setminus \{x\}$ that gets colour~$2$. If $X=\{x\}$, we will not try to find this vertex; for our algorithm its existence will suffice.

By Property~{\bf P2}, every vertex $x^*\in X^*$ has list $L(x^*)\subseteq \{1,3\}$. 
By {\bf P2}, we find that $x^*$ has a neighbour $x'\in X$, which has $L'(x')=\{2\}$. By the assumption of Case~2b, there exists at least one other vertex~$x''$ in $N_1$ that gets colour~$2$. Then we cannot give $x^*$ colour~$1$, as otherwise $x'',v,x',x^*$ would form a bichromatic $P_4$.

Due to the above, we can remove colour~$1$ from the list of every vertex of $X^*$ and afterwards we have $L(x^*)=\{3\}$ for every $x^*\in X^*$. 
We now remove colour~$3$ from the list of every neighbour of a vertex of $X^*$. As $L'$ is a $2$-list assignment that does not assign any vertex of $G'$ the list $\{1,2\}$, we find afterwards that every neighbour of every vertex of $X^*$ in $N_3$ has list $\{1\}$ or $\{2\}$. Moreover, it follows that $X^*$ is an independent set (as otherwise we discard $(G',L')$).
No vertex of $W^*\cup Y^*$ is adjacent to any vertex in $X^*$ (by Property~{\bf P2}).
Hence, every vertex in $X^*$ has no neighbours in $N_2$.

We now prove that no vertex in $S_2$ can receive colour~$3$. For contradiction, assume that $c$ is a star $3$-colouring of $G$ that respects $L'$ and that assigns a vertex $s\in S_2$ colour $c(s)=3$. As $G'$ has diameter~$3$, there is  a path~$P$ from $s$ to $x\in X$ of length at most~$3$. Then $P$ is of the form $s-N_2-x$ or $s-N_3-N_2-x$ or $s-N_2-N_2-x$ or $s-N_2-N_1-x$.
If $P$ is of the form $s-N_2-x$, then $s$ has a neighbour in $X^*$, which has list $\{3\}$. Hence, as $s$ received colour~$3$, this is not possible. We show that the other three cases are not possible either.

First suppose that $P$ is of the form $s-N_3-N_2-x$, say $P= szx^*x$ for some $z\in N_3$ and $x^*\in N_2$. As no vertex of $W^*\cup Y^*$ is adjacent to any vertex in $X$, we find that $x^*\in X^*$. This means that $z$ must receive colour~$1$, as otherwise the vertices $x$, $x^*$, $z$, $s$ would form a bichromatic $P_4$. As $s\in S_2$, we find that $s$ has two neighbours $t_1$ and $t_2$ in $T_2$. Both $t_1$ and $t_2$ have list $\{1,3\}$, so they must receive colour~$1$. At least one of them, say $t_1$, is not equal to $z$. However, now $x^*$, $z$, $s$, $t_1$ form a bichromatic $P_4$, a contradiction. Hence, this case cannot happen.

Now suppose that $P$ is of the form $s-N_2-N_2-x$, say $P=szx^*x$ for some $z,x^*\in N_2$. As no vertex of $W^*\cup Y^*$ is adjacent to any vertex in $X$, $x^*\in X^*$. However, no vertex in $X^*$ has a neighbour in $N_2$. Hence, this case cannot happen.
 
Finally, suppose that $P$ is of the form $s-N_2-N_1-x$, say $P=sw^*wx$ for some $w^*\in N_2$ and $w\in N_1$. As $X$ is independent and no vertex of $Y$ is adjacent to a vertex of $X$, we find that $w\in W$ and thus $w^*\in W^*$. However, this is not possible, as $s\in S_2$ is not adjacent to any vertex in $W^*$ by definition. Hence, this case cannot happen either, so we have proven the claim. So, we can remove colour~$3$ from the list of every vertex $s\in S_2$. Hence, $L'(s)=\{1\}$ for every $s\in S_2$. 

We construct $G'_s$ and $L'_s$ in $O(m^2)$ time by Lemma~\ref{l-ex}. We claim that $L'_s$ is a $2$-list assignment of $G'_s$. This can be seen as follows. Let $e=ab$ be an unsuitable edge of $G'$. As $G'$ has no vertices with list $\{1,2\}$, we may assume that $L'(a)=\{1,3\}$ and $L'(b)=\{2,3\}$. As every vertex in $R$ is adjacent to a vertex in $X^*$ with list $\{3\}$, no vertex in $R$ has list $\{1,3\}$. We just deduced that no vertex in $S_2$ has list $\{1,3\}$ either. Hence, the only vertices with list $\{1,3\}$ belong to $T_2$, so $a\in T_2$. Then, by definition, we find that $a$ has a neighbour $w\in W^*$, which has a neighbour $w\in W$. As $w^*$ has list $\{2\}$ and $w$ has list $\{3\}$, 
the edge $w^*w$ is list-reducing. Hence, $a$ is a fixer for the edge $ab$. The claim now follows from Lemma~\ref{l-ex}, and by the same lemma, it remains to check if $G'_s$ has a $3$-colouring that respects $L'_s$. We can do the latter in $O(m^2)$ time by Theorem~\ref{t-2sat}.

\smallskip
\noindent
This concludes the description of our algorithm.
Its correctness follows from the correctness of our branching steps. The total running time is $O(nm^2)$, as there are $O(n)$ branches, and we can deal with each branch in $O(m^2)$ time. \qed
\end{proof}

\noindent
We complement the previous, algorithmic result by a hardness result, which is just an observation on a known construction~\cite{ACKKR04}.

\begin{lemma}\label{lem:hard-star}
{\sc Star $3$-Colouring} is \NP-complete on graphs of diameter at most~$8$.
\end{lemma}

\begin{proof}
We recall the Albertson et al.~\cite{ACKKR04} proved that {\sc Star $3$-Colouring} is \NP-complete by making the following reduction from {\sc $3$-Colouring}, which is \NP-complete even for graphs of diameter~$3$~\cite{MS16}.
Let $G$ be a graph of diameter~$3$.
For each $uv$ do as follows. Remove $uv$ and make both $u$ and $v$ adjacent to three new vertices $x_{uv}$, $y_{uv}$ and $z_{uv}$. Then
$G$ has $3$-colouring if and only if the new graph $G'$ has a star $3$-colouring~\cite{ACKKR04}. It remains to observe that $G'$ has diameter at most~$8$. \qed
\end{proof}

\noindent
We are now ready to prove Theorem~\ref{t-main2}.

\medskip
\noindent
{\bf  Theorem~\ref{t-main2} (restated).}
{\it For $d\geq 1$ and $k\geq 3$,  
{\sc Star $k$-Colouring} on graphs of diameter at most~$d$ is
polynomial-time solvable if $d=1$, $k\geq 4$ or $d\leq 3$, $k=3$ and \NP-complete if $d\geq 2$, $k\geq 4$ or $d\geq 8$, $k=3$.}

\begin{proof}
The cases $d\leq 3$, $k=3$ and $d\geq 8$, $k=3$ follow from Lemmas~\ref{l-sd2} and~\ref{lem:hard-star}, respectively. 
The case $d=1$, $k\geq 4$ is trivial. For the case $d\geq 2$, $k\geq 4$ we reduce from {\sc Star $3$-Colouring}: to an instance $G$ of {\sc Star $k$-Colouring}, we add a clique of $k-3$ vertices, which we make adjacent to every vertex of $G$. 
\end{proof}

\section{${\mathbf{L(1,2)}}$-Labelling for Graphs of Diameter~$\mathbf{2}$}\label{s-injective}

In this section we prove the missing case in Theorem~\ref{t-main3}, namely that {\sc $L(1,2)$-Labelling} is \NP-complete even for graphs of diameter~$2$. We need three lemmas.
We first present, as Lemma~\ref{l-1} and~\ref{l-2}, two hardness results for {\sc Hamiltonian Cycle}.
We use Lemma~\ref{l-1} to prove Lemma~\ref{l-2}, and the latter to prove Lemma~\ref{l-final}.

The {\it eccentricity} of a vertex $u$ in a graph is the maximum distance of $u$ to some other vertex of $G$. The {\it radius} of $G$ is the minimum eccentricity of $G$.

\begin{lemma}\label{l-1}
{\sc Hamiltonian Cycle} is \NP-complete even for connected bipartite graphs of minimum degree~$2$ and maximum degree~$5$ that have the following three additional properties:\\[-10pt]
\begin{enumerate}
\item for every two vertices $x,y$ that belong to the same partition class and that have no common neighbour, there exists a vertex in the same partition class as $x,y$ that is of distance greater than~$2$ from both $x$ and $y$;\\[-10pt]
\item for every two non-adjacent vertices $x, y$ that belong to different partition classes, either $x$ has a neighbour of distance greater than $2$ from $y$, or $y$ has a neighbour of distance greater than $2$ from $x$, and\\[-10pt]
\item no two vertices of degree~$2$ have the same neighbourhood.
\end{enumerate}
\end{lemma}

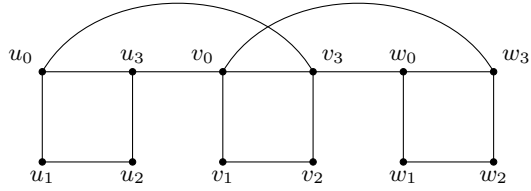
\begin{figure}
\begin{center}
\begin{tikzpicture}[scale=0.6]
\draw[fill=black] (0,0) circle[radius=2pt] ;
\draw[fill=black](2,0) circle[radius=2pt];
\draw[fill=black](0,2) circle[radius=2pt];
\draw[fill=black](2,2) circle[radius=2pt];
\node[below] at (0,0){$u_1$};
\node[below] at (2,0){$u_2$};
\node[above left] at (0,2){$u_0$};
\node[above] at (2,2){$u_3$};
\draw[fill=black] (4,0) circle[radius=2pt] ;
\draw[fill=black](6,0) circle[radius=2pt];
\draw[fill=black](4,2) circle[radius=2pt];
\draw[fill=black](6,2) circle[radius=2pt];
\node[below] at (4,0){$v_1$};
\node[below] at (6,0){$v_2$};
\node[above left] at (4,2){$v_0$};
\node[above right] at (6,2){$v_3$};
\draw[fill=black] (8,0) circle[radius=2pt] ;
\draw[fill=black](10,0) circle[radius=2pt];
\draw[fill=black](8,2) circle[radius=2pt];
\draw[fill=black](10,2) circle[radius=2pt];
\node[below] at (8,0){$w_1$};
\node[below] at (10,0){$w_2$};
\node[above] at (8,2){$w_0$};
\node[above right] at (10,2){$w_3$};
\draw(4,2)--(6,2);
\draw(4,0)--(4,2);
\draw(6,0)--(6,2);
\draw(4,0)--(6,0);
\draw(0,2)--(2,2);
\draw(0,0)--(0,2);
\draw(2,0)--(2,2);
\draw(0,0)--(2,0);
\draw(8,2)--(10,2);
\draw(8,0)--(8,2);
\draw(10,0)--(10,2);
\draw(8,0)--(10,0);
\draw(2,2)--(4,2);
\draw(0,2) to[out=60, in=120](6,2);
\draw(4,2) to[out=60, in=120](10,2);
\draw(6,2) --(8,2);
\end{tikzpicture}
\end{center}
\caption{The graph $G^{\prime}$ from the proof of Lemma~\ref{l-1}, when $G$ is the $3$-vertex path $uvw$.}\label{f-p3}
\end{figure}

\begin{proof}
We reduce from {\sc Hamiltonian Cycle}, which  is \NP-complete even for graphs of maximum degree~$3$~\cite{GJT76}. As graphs of bounded maximum degree and bounded radius have constant size, 
the problem remains \NP-complete if in addition we assume that the input graph~$G=(V,E)$ of maximum degree~$3$ has radius at least~$10$.

We follow the construction used in~\cite{Kr75}. That is, from $G$ we construct a graph $G^{\prime}=(V^{\prime}, E^{\prime})$ as follows. We replace each $v\in V$ by a $4$-cycle $v_0, v_1,v_2,v_3$. Moreover, for each $uv\in E$, we do as follows. Let $u_0,u_1,u_2,u_3$ and $v_0,v_1,v_2,v_3$ be the $4$-cycles that are associated with $u$ and $v$, respectively. 
We add the two edges $u_0v_3$ and $u_3v_0$. This gives us the graph~$G^\prime$. See also Figure~\ref{f-p3}. It is readily seen that $G$ has a Hamiltonian cycle if and only if $G^\prime$ has a Hamiltonian cycle.
Moreover,	 $G^{\prime}$ is bipartite with one part $A= \{v_i: i=0,2 \}$ and the other $B= \{v_i: i=1,3\}$, and $G^\prime$ has minimum degree~$2$ and maximum degree~$5$; the latter holds as every vertex $v_i$ has two more neighbours than $v$ and $v$ has degree at most~$3$ (as $G$ has maximum degree~$3$). We now prove properties 1--3.

We first show Property~1. Let $x$ and $y$ be in the same partition class, say $A$, and assume that $x$ and $y$ have no common neighbour.
If every vertex of $A$ is of distance~$2$ from either $x$ or $y$ then, as $G$ is connected, $x$ and $y$ are of distance at most~$6$ from each other. Consequently, the distance from $x$ to any other vertex is at most $6+2+1=9$.
Hence, $G^\prime$ has radius at most $9$. As the distance between any two vertices $u_i$ and $v_i$ in $G^\prime$ is at least the distance between $u$ and $v$ in $G$, we find that  $G$ also has radius at most $9$, a contradiction.

We now show Property~$2$. Let $x\in A$ and $y\in B$ be non-adjacent. Then $x=u_i$ for some $i\in \{0,2\}$ and $y=v_j$ for some $j\in \{1,3\}$ for vertices $u,v\in V$ with $u \neq v$. First suppose that $x=u_0$.
If $y=v_1$, then $u_1$ is adjacent to $u_0$ and shares no neighbour with $v_1$, since $u \neq v$. If $y=v_3$ then $v_2$ is adjacent to $v_3$ and shares no neighbour with $u_0$, since $x=u_0$ and $y=v_3$ are non-adjacent. Now suppose that $x=u_2$. If $y=v_1$, then $u_1$ is adjacent to $u_2$ and shares no neighbour with $v_1$. Finally, if $x=u_2$ and $y=v_3$ then $v_2$ is adjacent to $v_3$ and shares no neighbour with $u_2$.
	 
Finally, Property~3 holds since the set of vertices of degree~$2$ is $\{v_1, v_2: v \in V\}$, and no pair of vertices from this set has the same neighbours. \qed
\end{proof}

\begin{figure}
\begin{center}
\begin{tikzpicture}[scale=0.6]
\draw[fill=black](-2,0) circle [radius=2pt];
\draw[fill=black](-1,0) circle [radius=2pt];
\draw[fill=black](-3,0) circle [radius=2pt];
\draw[fill=black] (0,0) circle[radius=2pt] ;
\draw[fill=black](2,0) circle[radius=2pt];
\draw[fill=black](0,2) circle[radius=2pt];
\draw[fill=black](2,2) circle[radius=2pt];
\node[below right] at (-1,0){$x_1$};
\node[below right] at (0,0){$x$};
\node[below] at (-3.1,0){${x_{1}^{\prime}}$};
\node[below] at (-2.1,0) {$x^{\prime}$};
\draw (-2,0)--(0,2);
\draw (-1,0)--(0,0);
\draw (-2,0) to [out=300, in=240] (2,0);
\draw[fill=black] (4,0) circle[radius=2pt] ;
\draw[fill=black](6,0) circle[radius=2pt];
\draw[fill=black](4,2) circle[radius=2pt];
\draw[fill=black](6,2) circle[radius=2pt];
\draw[fill=black] (8,0) circle[radius=2pt] ;
\draw[fill=black](10,0) circle[radius=2pt];
\draw[fill=black](8,2) circle[radius=2pt];
\draw[fill=black](10,2) circle[radius=2pt];
\draw(4,2)--(6,2);
\draw(4,0)--(4,2);
\draw(6,0)--(6,2);
\draw(4,0)--(6,0);
\draw(0,2)--(2,2);
\draw(0,0)--(0,2);
\draw(-3,0)--(-2,0);
\draw(2,0)--(2,2);
\draw(0,0)--(2,0);
\draw(8,2)--(10,2);
\draw(8,0)--(8,2);
\draw(10,0)--(10,2);
\draw(8,0)--(10,0);
\draw(2,2)--(4,2);
\draw(0,2) to[out=60, in=120](6,2);
\draw(4,2) to[out=60, in=120](10,2);
\draw(6,2) --(8,2);
\end{tikzpicture}
\end{center}
\caption{The graph $G''$ from the proof of Lemma~\ref{l-2}.}\label{f-f2}
\end{figure}
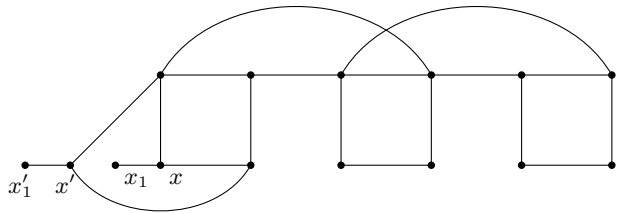

\begin{lemma}\label{l-2}
{\sc Hamiltonian Path} is \NP-complete even for connected bipartite graphs that satisfy Properties 1 and 2 of Lemma~\ref{l-1}.
%the following two properties (which are the same as Properties 1 and 2 of Lemma~\ref{l-1}):\\[-10pt]
%\begin{enumerate}
%\item for every two vertices $x,y$ that belong to the same partition class and that have no common neighbour, there exists a vertex in the same partition class as $x,y$ that is of distance greater than~$2$ from both $x$ and $y$;\\[-10pt]
%\item for every two non-adjacent vertices $x, y$ that belong to different partition classes, either $x$ has a neighbour of distance greater than $2$ from $y$, or $y$ has a neighbour of distance greater than $2$ from $x$.\\[-10pt]
%\end{enumerate}
\end{lemma}

\begin{proof}
We reduce from {\sc Hamiltonian Cycle}, which is \NP-complete even for the graphs $G^{\prime}=(V^{\prime}, E^{\prime})$ constructed in the proof of Lemma~\ref{l-1}. We modify a given graph $G^\prime$ into a graph $G''$ as follows. We take some vertex $x$ of degree $2$ and add a new vertex $x^{\prime}$ with the same neighbourhood as $x$. We then add two further new vertices, $x_1$ and $x^{\prime}_1$ such that $x_1$ is pendant on $x$ and $x^{\prime}_1$ is pendant on $x^{\prime}$. See also Figure~\ref{f-f2}.  We observe that $G^\prime$ has a Hamiltonian cycle if and only if $G''$ has a Hamiltonian path, which must start in $u_1$ and end in $u_1'$. As $G^\prime$ is bipartite, $G''$ is also bipartite. Hence, it remains to prove Properties~1 and~2.
	
We first show that Property~1 holds for $G''$.	As Property~1 holds for $G^\prime$ by Lemma~\ref{l-1} and the three new vertices $x',x_1,x_1'$ do not decrease the distance between any two vertices of $G^\prime$,
we only need to consider pairs of vertices involving at least one of $\{x',x_1,x_1'\}$. Vertices $x_1$ and $x_1'$ belong to the same partition class of $G''$ and have no common neighbour. Any non-neighbour $z$ of $x$ in $G^\prime$ is of distance greater than~$2$ from both $x_1$ and $x_1'$, and we can choose $z$ such that $z$ belongs to the same partition class of $G''$ as $x_1$ and $x_1'$. Now consider one of $x_1,x_1'$, say $x_1$, and a vertex $y$ of $G^\prime$ that belongs to the same partition class as $x_1$ in $G''$, such that $x_1$ and $y$  do not have a common neighbour. Then $x_1'$ is of distance greater than~$2$ from $y$ in $G''$, and we can take $x_1'$.
Vertices $x$ and $x'$ also belong to the same partition class of $G''$, but their neighbourhood is the same. Therefore, as Property~1 holds with respect to $x$ in $G^\prime$, Property~1 also holds with respect to $x'$ in $G''$.
	
We now show that Property 2 holds for $G''$. Again we need only to verify pairs involving at least one of $\{x',x_1,x_1'\}$. We first consider the pair $(x',x_1)$; note that $x'$ and $x_1$ are non-adjacent and belong to different partition classes of $G''$. We can take $x_1'$ as the desired vertex, as $x_1'$ is adjacent to $x'$ but of distance greater than~$2$ from $x_1$ in $G''$. By symmetry, Property~2 holds for the pair $(x,x_1')$. 

We now consider a pair $(x',y)$ where
$y\notin \{x_1,x_1'\}$  belongs to a different partition class of $G''$ than $x'$ and is not adjacent to $x'$. As $x$ and $x'$ have the same neighbourhood in $G''$, we find that $y$ and $x$ are non-adjacent vertices in different partition classes as well. As Property~2 holds for $G^\prime$, there exists a vertex $z$ that is a neighbour of one of $\{x,y\}$ but that is of distance greater than~$2$ from the other vertex of $\{x,y\}$. As the distance between two vertices of $G^\prime$ is the same in $G''$, we can take $z$ as the desired vertex for the pair $(x',y)$.

Finally, we consider a pair $(x_1,y)$ or $(x_1',y)$, say $(x_1,y)$ (by symmetry), where $y$ is a non-neighbour of $x_1$ in $G''$ such that $x_1$ and $y$ belong to different partition classes of $G''$. Note that $y$ must be a vertex of $G^\prime$. For contradiction, assume that
every neighbour of $y$ is of distance~$2$ from $x_1$ in $G''$. Then every neighbour of $y$ in $G''$ is a neighbour of $x$. As $y$ belongs to $G^\prime$, we find that $y$ has degree at least~$2$ in $G^\prime$. As $x$ has degree~$2$ in $G^\prime$, this means that in $G^\prime$, both $x$ and $y$ have the same neighbourhood. The latter is a contradiction, as $G^\prime$ satisfies Property~3 of Lemma~\ref{l-2}.
We conclude that $G''$ has Property~2. \qed
\end{proof}	

\begin{lemma}\label{l-final}
It is \NP-complete to decide if a graph has a Hamiltonian path, no edge of which is contained in a triangle, even for graphs of diameter~$2$.
\end{lemma}

\begin{proof}
We reduce from {\sc Hamiltonian Path}, which is \NP-complete even for the graphs $G''=(V'',E''')$ constructed in the proof of Lemma~\ref{l-2}. We modify a given graph $G''$ into a graph $G^*$ by adding an edge between any two vertices $u, v$ that belong to the same partition class and that are of distance greater than~$2$ from each other in $G''$.
By our construction, the distance between any two vertices that belong to the same partition class of $G''$ is at most~$2$ in $G^*$. As $G''$ has Property~2, the distance between any two vertices in different partition classes of $G''$ is at most~$2$ in $G^*$ as well. Hence, $G^*$ has diameter at most~$2$.

It remains to prove that $G''$ has a Hamiltonian path if and only if $G^*$ has a Hamiltonian path, no edge of which is contained in a triangle. For showing this it suffices to prove that for every edge $e$ of $G^*$, it holds that $e$ does not belong to a triangle in $G^*$ if and only if $e$ is an edge of $G''$. 

First suppose that $e$ is not an edge of $G''$. Say $e$ is an edge between $x$ and $y$, where $x$ and $y$ are two vertices of distance greater than $2$ that belong to the same partition class of $G''$. As $G''$ has Property~1, there exists a vertex $z$ that also belongs to the same partition class as $x$ and $y$ and that is of distance greater than $2$ from both $x$ and $y$. Hence, we have added the edges $xz$ and $yz$ as well, thus $e$ belongs to a triangle in $G^*$.

Now suppose that $e$ is an edge of $G''$. Let $e=xy$ for two vertices $x$ and $y$ (which belong to different bipartition classes of $G''$). For contradiction, assume that $x$ and $y$ are contained in a triangle $xyz$ where $z$ belongs to the same partition class as $x$, so we added the edge~$xz$. Note that $x$ and $z$ have a common neighbour in $G''$, namely $y$. This means that their distance is not greater than~$2$ in $G''$. Hence, we would not have added the edge $xz$, a contradiction. \qed
\end{proof}

\noindent
We can now prove our main result. For doing this, we show that an $n$-vertex graph $G$ of diameter~$2$ has an $L(1,2)$-$n$-labelling if and only if $G$ has a Hamiltonian path, no edge of which is contained in a triangle.  

\begin{theorem}\label{t-hard}
The {\sc $L(1,2)$-Labelling} problem is \NP-complete even for graphs of diameter at most~$2$.
\end{theorem}

\begin{proof}
Let $G$ be an $n$-vertex graph of diameter~$2$. It suffices to prove that $G$ has an $L(1,2)$-$n$-labelling if and only if $G$ has a Hamiltonian path, no edge of which is contained in a triangle. Then, afterwards, we can apply Lemma~\ref{l-final}.

First suppose that $G$ has an $L(1,2)$-$n$-labelling $c$. Since $G$ has diameter~$2$, any two non-adjacent vertices have a common neighbour. Hence, colours of non-adjacent vertices must differ by at least~$2$. Consequently, two vertices with consecutive colours must be adjacent. As colours of adjacent vertices differ by at least~$1$, we also find that no two vertices have the same colour. Consequently, every colour~$i$ with $1\leq i\leq n$ is used.
Therefore we have a Hamiltonian path $P=v_1 \dots v_n$ where $v_i$ is the vertices with colour~$c(v_i)=i$. No edge $v_iv_{i+1}$ is contained in a triangle since there can be no path of length~$2$ between $v_i$ and $v_{i+1}$.

Now suppose that $G$ contains a Hamiltonian path $P=v_1 \dots v_n$, no edge of which is contained in a triangle. The latter means that there is no path of length~$2$ between $v_i$ and $v_{i+1}$ for $i\in \{1,\ldots,n-1\}$.
Then we obtain an $L(1,2)$-$n$-labelling $c$ by defining $c(v_i)=i$. \qed
\end{proof}

\section{Conclusions}\label{s-con}

We obtained (almost) complexity dichotomies for classical variants of the graph colouring problem by bounding the diameter of the graph. 
In particular, we proved that {\sc Acyclic $3$-Colouring} is polynomially solvable for graphs of diameter at most~$2$ and that for {\sc Star $3$-Colouring} this holds even for graphs of diameter at most~$3$. We are not aware of any other problems
that are polynomial-time solvable on graphs of diameter at most~$3$ but \NP-complete on graphs of diameter $d$ for some $d>3$.
In light of this it would be interesting to
close the gaps in Theorems~\ref{t-main1} (one open case) and~\ref{t-main2} (four open cases). This seems challenging. 
 The \NP-hardness construction of Mertzios and Spirakis~\cite{MS16} for {\sc $3$-Colouring} of graphs of diameter~$3$ does lead to \NP-hardness for {\sc Near Bipartiteness} for graphs of diameter~$3$, as observed by Bonamy et al.~\cite{BDFJP18}. However, the construction of~\cite{MS16} cannot be used for {\sc Acyclic $3$-Colouring} and {\sc Star $3$-Colouring}. Hence, new techniques are required.

\bibliographystyle{splncs04}

\end{document}